\documentclass[journal,twocolumn]{IEEEtran}

\usepackage{amsfonts}
\usepackage{times}
\usepackage[pdftex]{graphicx}
\usepackage{latexsym}
\usepackage{dsfont}
\usepackage{amssymb}
\usepackage{amsmath}
\usepackage{cite}
\usepackage{verbatim}
\usepackage{caption}
\usepackage{algorithmic}
\usepackage[margin=5mm]{subcaption}
\usepackage[outdir=./]{epstopdf}
\newtheorem{theorem}{Theorem}
\usepackage[margin=2.0cm]{geometry}
\usepackage{cite}
\usepackage[pdftex]{graphicx}
\usepackage{psfrag}
\usepackage{amsmath}
\usepackage{amssymb}
\usepackage{epsfig}
\usepackage{color}
\usepackage{epsf}
\usepackage{setspace}
\usepackage{caption}
\usepackage{subcaption}
\usepackage{fancyhdr}
\usepackage{bbm}

\newtheorem{lemma}[theorem]{Lemma}

\newenvironment{proof}{ \textbf{Proof:} }{ \hfill $\Box$}

\newcommand{\figref}[1]{{Fig.}~\ref{#1}}

\def\bb0{{\mathbb{0}}}

\def\bb{{\mathbf{b}}}

\def\b0{{\mathbf{0}}}

\def\bbE{{\mathbb{E}}}

\def\bbI{{\mathbb{I}}}

\def\bbP{{\mathbb{P}}}

\def\bbR{{\mathbb{R}}}

\def\cF{\mathcal{F}}

\def\cL{\mathcal{L}}

\def\sf0{{\mathsf{0}}}

\begin{document}

\title{Performance Analysis of mmWave Ad Hoc Networks}
\author{Andrew Thornburg, Tianyang Bai, and Robert W. Heath Jr.
\thanks{A. Thornburg, T. Bai, and R. W. Heath Jr. are with the Wireless Networking and Communications Group at The University of Texas at Austin, Austin, TX, USA (email: \{andrew.thornburg,tybai,rheath\}@utexas.edu). Parts of this paper appeared at the IEEE 2015 International Conference on Communications and the IEEE 2015 International Conference on Acoustics, Speech, and Signal Processing. This work was supported in part by Army Research Labs under Grant No. W911NF-12-R-0011 and the National Science Foundation under Grant No. 1218338.}}

\maketitle

\begin{abstract}
\emph{Ad hoc} networks provide an on-demand, infrastructure-free means to communicate between soldiers in war zones, aid workers in disaster areas, or consumers in device-to-device (D2D) applications. Unfortunately, \emph{ad hoc} networks are limited by interference due to nearby transmissions. Millimeter-wave (mmWave) devices offer several potential advantages for \emph{ad hoc} networks including reduced interference due to directional antennas and building blockages, not to mention huge bandwidth channels for large data rates.. This paper uses a stochastic geometry approach to characterize the one-way and two-way signal-to-interference ratio distribution of a mmWave \emph{ad hoc} network with directional antennas, random blockages, and ALOHA channel access. The interference-to-noise ratio shows that a fundamental limitation of an \emph{ad hoc} network, interference, may still be an issue. The performance of mmWave \emph{ad hoc} networks is bounded by the transmission capacity and area spectral efficiency. The results show that mmWave networks can support much higher densities and larger spectral efficiencies, even in the presence of blockage, compared with lower frequency communication for certain link distances. Due to the increased bandwidth, the rate coverage of mmWave can be much greater than lower frequency devices.
\end{abstract}
%%%%%%%%%%%%%%%%%%%%%%%%%%%%%%%%%%%%%%%%%%%%%%
\section{Introduction} \label{sec:intro}
%%%%%%%%%%%%%%%%%%%%%%%%%%%%%%%%%%%%%%%%%%%%%%
 
Next-generation \emph{ad hoc} networks, such as military battlefield networks, high-fidelity emergency response video, or device-to-device (D2D) entertainment applications, must offer high data rates and ubiquitous coverage. Typically, \emph{ad hoc} networks are limited by the uncoordinated interference created by proximate transmitters. Measurement studies and analysis of indoor, commercial wireless PAN/LAN systems show that mmWave systems may experience less interference due to directional antennas and building blockage in addition to offering massive bandwidth \cite{RapSun13,ZhuDou11,SinMud11,RapHeaBook,Rappaport2015}. While these results are promising, the potential performance of outdoor mmWave \emph{ad hoc} networks incorporating key features like directional antennas and building blockage is not yet understood. 

Prior work has leveraged stochastic geometry to calculate the performance of \emph{ad hoc} networks \cite{WebAnd10}. The transmission capacity is a information theoretic performance metric calculated using stochastic geometry \cite{AndSha08, Hunter08, HuaAnd12, WebAnd10}. The transmission capacity is the maximum spatial density (users per m$^2$) of transmitters given an outage constraint and is well studied, see \cite{VazHea12,VazTru11,WebAnd10}, and references therein. A related metric is the area spectral efficiency which yields the bits/sec/Hz/m$^2$ of a network \cite{AndGan10}. Both metrics are widely used to compare and contrast transmission techniques and network architectures.

%In \cite{VazHea12}, the transmission capacity using MIMO beamforming for \emph{ad hoc} networks was derived. A main conclusion is that the receive antennas should be used for interference cancellation not for spatial multiplexing. The ability to perform spatial interference cancellation, however, is limited in mmWave due to hardware constraints \cite{AlkElA13}; each antenna on a mmWave device will likely not have dedicated baseband hardware. 

Beamforming has been analyzed with stochastic geometry and other methods in \emph{ad hoc} networks under the term \emph{smart antennas}, phased arrays, or adaptive antennas. Prior work on \emph{ad hoc} networks considered smart antennas and other directional antennas \cite{Hunter08,RamRed05,BelFou02,Win06,ChoYanRam:On-designing-MAC-protocols-for-wireless:06}. The transmission capacity of \emph{ad hoc} networks with directional antennas was computed in \cite{Hunter08} assuming small-scale Rayleigh fading. A directional MAC testbed was benchmarked in \cite{RamRed05}. In \cite{BelFou02}, the analyses and performance of the system assumes Rayleigh fading. The optimization of the MAC for directional antennas was discussed in \cite{Win06,ChoYanRam:On-designing-MAC-protocols-for-wireless:06}. While the results are frequency agnostic, the results are built around channel models that reasonably apply only for sub-mmWave systems.

Blockage is an important impairment in mmWave \emph{ad hoc} systems. Work in \cite{miweba,Rappaport2015,RapSun13} showed that the path-loss models were different between line-of-sight (LOS) and non-line-of-signt (NLOS) due to building blockage. This was the basis of the stochastic geometry analysis in \cite{Bai2014} which was applied to cellular systems. The exclusion zone of the cellular system model in \cite{Bai2014} is not applicable to \emph{ad hoc} networks. In the cellular model, the users fall in the Voronoi cell of the base station. The strongest interferer due to the Voronoi structure must lie outside a ball centered at the receiver. While in an Aloha \emph{ad hoc} network, an interfering transmitter can be arbitrary close \cite{baccelli_app}. In \cite{Gowaikar2006}, blockage results from small-scale fading. At mmWave frequencies, blockages are due to obstacles like buildings which heavily attenuate mmWave signals \cite{Bai2013}. The effect of blockage is developed in \cite{Bai2014} for mmWave cellular networks; rate trends for cellular are captured with real-world building footprints in \cite{Mandar}. A LOS-ball approach is taken in \cite{Singh2015} for backhaul networks which is validated using real-world building data. Wearable networks which quantified the effect of human blockage was considered in \cite{Kiran}. No consideration has been made in the literature, however, to the effect of blockage on outdoor mmWave \emph{ad hoc} networks.

%There is considerable work with outdoor mmWave cellular networks analyzed in a stochastic geometry framework \cite{Bai2014a,Bai2013c,Bai2014,AkoElA12}. Cellular networks, however, have a specific exclusion region due to base station association that is not present in \emph{ad hoc} networks and are generally not as dense as \emph{ad hoc} networks. Previous work on mmWave \emph{ad hoc} networks was largely restricted to indoor scenarios with limited range; the new 802.11ad standard which supports multi-hop communication is based on mmWave technology for low mobility, indoor environments \cite{ZhuDou11}. The 802.15.3c and WirelessHD standards also operate in the mmWave bands; they are designed for \emph{ad hoc} WPAN operation but the usage models and design is based on short, indoor, and stationary communication ranges of less than 10m \cite{Baykas2011}.

In this paper, we formulate the performance of mmWave \emph{ad hoc} networks in a stochastic geometry framework. We incorporate random factors of a mmWave \emph{ad hoc} network such as building blockage, antenna alignment, interferer position, and user position. Using a similar framework, we compare and contrast the performance against a lower frequency UHF \emph{ad hoc} network. The main contributions of the paper can be summarized as follows:
\begin{itemize}
	\item Derivation of a bound for mmWave \emph{ad hoc} network signal-to-interference-and-noise (SINR) complimentary cumulative distribution function (CCDF). A version of the SINR proof for line-of-sight communication appeared in \cite{Thornburg2014}; we have, however, strengthened the proof to bound the result rather than approximate it as well as extend it to non-line-of-sight. Using the SINR CCDF, a Taylor approximation is used to compute the transmission capacity and area spectral efficiency of the network. We argue for LOS-aware protocols due to the large performance increase from LOS communication at mmWave. Lastly, we calculate the effect of random receiver location on performance.
	\item Computation of the interference-to-noise ratio (INR). The interference-to-noise ratio distribution (INR) derivation of \cite{Thornburg2014b} is similarly strengthened; additionally, we include discussion of the INR when a network is operating at the transmission capacity.
	\item Characterizing the effect of two-way communication on the transmission capacity and area spectral efficiency. We show that optimal bandwidth allocation leads to large gains in both performance metrics.
\end{itemize}

The rest of the paper is organized as follows. Section II provides the system model and assumptions used in the paper. Section III derives the SINR distribution, transmission capacity, ASE, and INR distribution for the one-way network. Section IV quantifies the transmission capacity and ASE for two-way networks. We present the results in Section V and conclude the paper in Section VI. Throughout the paper, $\bbP[X]$ is the probability of event $X$ and $\bbE$ is the expectation operator.

%%%%%%%%%%%%%%%%%%%%%%%%%%%%%%%%%%%%%%%%%%%%%%
\section{System Model}\label{sec:sys}
%%%%%%%%%%%%%%%%%%%%%%%%%%%%%%%%%%%%%%%%%

%\begin{figure}
%\centering
%\includegraphics[width=0.65\columnwidth]{dipoles_v4.pdf}
%	\caption{A dipole \emph{ad hoc} network with direction antennas and blockages. The typical receiver is encountering interference from other nearby transmitters. The link distances are fixed, and each link may be blocked by a building.}
%	\label{fig:dipoles}
%\end{figure}

\subsection{Network Model}
Consider an \emph{ad hoc} network where users act as transmitter or receiver.  We use the dipole model of \cite{baccelli_app} where each transmitter in the network has a corresponding receiver at distance $r$. The transmitters operate at constant power with no power control. The location of the transmitting users within the network are points from a homogeneous Poisson point process (PPP) $\Phi$ on the Euclidean plane $\bbR^2$ with intensity $\lambda_\text{t}$, which is standard for evaluating the transmission capacity of \emph{ad hoc} networks, see \cite{WebAnd10} and the references therein. We analyze performance at the \emph{typical} dipole pair at the origin. The performance of the typical dipole characterizes the network performance thanks to Slivnyak's Theorem \cite{baccelli_app}. The channel is accessed using a synchronized slotted Aloha-type protocol with parameter $p_\text{tx}$. During each block, a user transmits with probability $p_\text{tx}$ or remains silent with probability $(1-p_\text{tx})$. We condition on a fully outdoor network. We define the effective transmitting user density, used throughout the rest of the paper, as
\begin{equation}
\lambda = p_\text{tx}p_\text{out}\lambda_\text{t},\label{efflam}
\end{equation}
where $p_\text{out}$ is the probability a user is outdoors. A homogeneous PPP is perhaps overly simplistic, but we leave the investigation of mmWave \emph{ad hoc} networks modeled with non-homogeneous PPPs to future work. We leave the optimization of $p_{\text{tx}}$ to future work, but provide a framework to find the solution in Section IV. The analysis of \cite{Bai2013} shows how to compute $p_\text{out}$ using stochastic geometry.

\subsection{Use of Beamforming}
Now we explain the role of beamforming in the mmWave signal model. The natural approach to combat increased omni-directional path-loss of mmWave is to use a large antenna aperture, which is achieved using multiple antennas \cite{AkoElA12,RapSun13,Pi2011}. The resulting array gain overcomes the frequency dependence on the path-loss and allows mmWave systems to provide reasonable link margin.  We denote the path-loss intercept as $A = 20\text{log}_{10}\left(\frac{2\pi d_\text{ref}}{\lambda_\text{ref}}\right)$ with $d_\text{ref} = 1m$ \cite{Rappaport2015} and $\lambda_\text{ref}$ as the carrier wavelength.

We assume that adaptive directional beamforming is implemented at both the transmitter and receiver where a main lobe is directed towards the dominate propagation path while smaller sidelobes direct energy in other directions. No attempt is made to direct nulls in the pattern to other receivers \cite{AkoKou11}; this is an interesting problem for future work. To facilitate the analysis, we approximate the actual beam pattern using a sectored model, as in \cite{Hunter08}. The beam pattern, $G_{\theta,M,m}$, is parameterized by three values: main lobe beamwidth ($\theta$), main lobe gain ($G$), and back lobe gain ($g$). Such an antenna is shown in Fig. \ref{fig:antenna_pic} where the mainlobe is 90$^\circ$, 30$^\circ$, or 9$^\circ$ with gain of 3dB, 10dB, or 15dB, respectively. The interferers are also equipped with directional antennas. Because the underlying PPP is isotropic in $\bbR^2$, we model the beam-direction of the typical node and each interfering node as a uniform random variable on $[0,2\pi]$. Thus, the effective antenna gain of the interference seen by the typical node is a discrete random variable described by
\begin{equation}
G_i = 
\begin{cases}
GG & \text{w.p.} \hspace{2mm} p_{\text{GG}}=(\frac{\theta}{\pi})^2 \\
Gg & \text{w.p.} \hspace{2mm} p_{\text{Gg}}=2\frac{\theta}{\pi}\frac{\pi-\theta}{\pi} \\
gg & \text{w.p.} \hspace{2mm} p_{\text{gg}}=(\frac{\pi-\theta}{\pi})^2 \\
\end{cases}.\label{eq:ant_gain}
\end{equation}

The typical dipole performs perfect beam alignment and thus has an antenna gain of $GG$. We note that the sectored model is pessimistic with regards to side band power. A typical uniform linear array, for instance, will consist of a main-lobe and many less powerful side-lobes each separated by nulls. The sectored model takes the most powerful side-lobe as the entire side-lobe (i.e. on average, the sectored model provides higher side-lobe power). Other work ignores the side-lobe power \cite{SinMud11}. 

 \begin{figure}
 	\centering
 	\includegraphics[width=\linewidth]{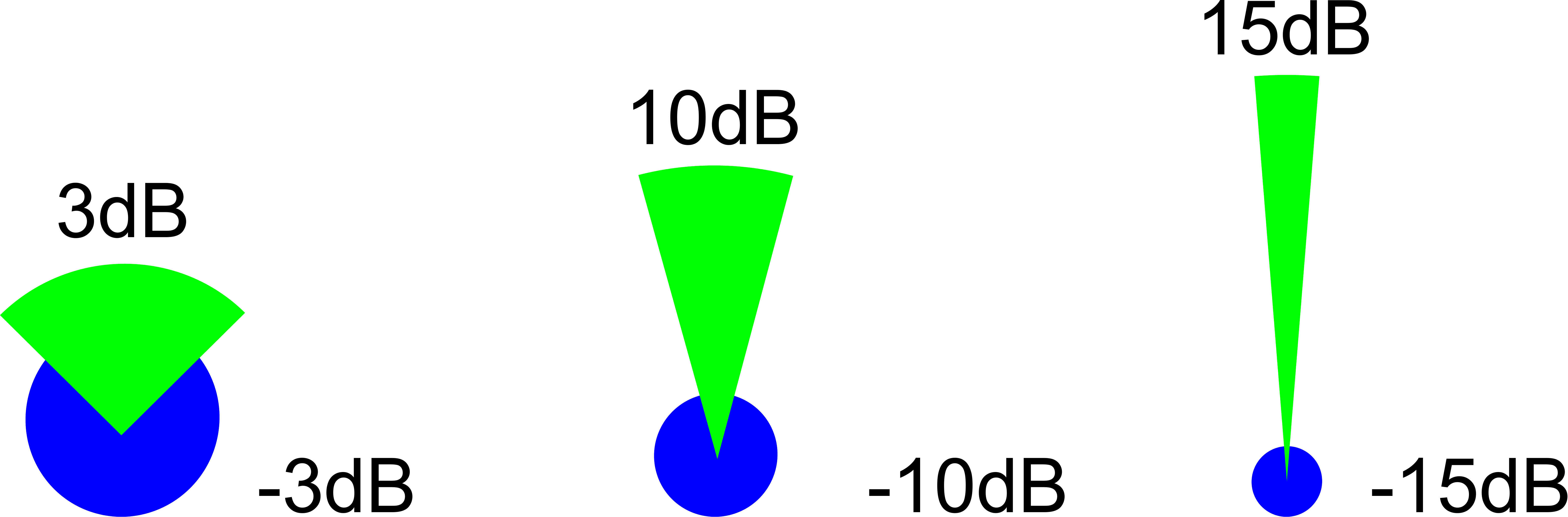}
 	\caption{An illustration of the sectored antenna model we use. Beamwidths are 90$^\circ$, 30$^\circ$, and 9$^\circ$, respectively.}
 	\label{fig:antenna_pic}
 \end{figure} 

\subsection{Blockage Model}
The signal path can be either unobstructed/LOS or blocked/NLOS, each with a different path-loss exponent. This distinction is supported by empirical measurements conducted in Austin, Europe, and Manhattan \cite{RapBen12,RapSun13, Rappaport2015,miweba}. The measurements conducted by \cite{Rappaport2015} include various vertical heights such as building (e.g. 17m) and closer-to-pedestrian (e.g. 7m). We believe the 7m measurements to be applicable to \emph{ad hoc} networks. The measurements of \cite{Rappaport2015}, at 28, 38, 60, and 73GHz, show the path-loss difference of LOS/NLOS. Additionally, a European consortium, MiWeba, has also conducted peer-to-peer urban canyon measurements made similar conclusions \cite{miweba}. One reason for larger difference in LOS/ NLOS path losses is that diffraction becomes weaker in mmWave, as the carrier frequency goes high \cite{miweba}. Besides, the Fresnel zone, whose size is proportional to the square of the wavelength, becomes smaller at mmWave. Therefore, the mmWave signals will be less likely affected by objects in the LOS links, and transmit as in free space \cite{miweba}. The work of \cite{Bai2013} assumes no particular architecture for the 2-dimensional stochastic geometry derivation. The work captures the distribution and placement of buildings with potential applications to cellular networks and \emph{ad hoc} networks.

The blockages are modeled as another Poisson point process of buildings independent of the communication network. Each point of the building PPP is independently marked with a random width, length, and orientation. Under such a scenario, it was shown that by using a random shape model of buildings to model blockage \cite{Bai2012,Bai2013}, the probability that a communication link of outdoor users is LOS is $\bbP[\text{LOS}] = e^{-\beta d} $,where $d$ is the link length and 
\begin{equation}
\beta = \frac{2\lambda_\text{b} (\bbE[W]+\bbE[L])}{\pi},\label{eq:beta}
\end{equation}
with $\lambda_\text{b}$ as the building PPP density, $\bbE[W]$ and $\bbE[L]$ are the average width and length, respectively, of the buildings.We note that the work of \cite{Bai2013} includes a parameter to capture the setting where transmitters are indoors, but this is not required in our model as we analyze outdoor networks and therefore condition on outdoor transmitters. A different analysis would be required for indoor networks. This is reasonable because because mmWave signals are heavily attenuated by many common building materials \cite{RapSun13}. For example, brick exhibits losses of 30dB at 28 GHz. While the leakage of indoor mmWave signals might be possible through open windows, we ignore the potential interference from indoors and focus solely on the outdoor setting.

The path-loss exponent on each interfering link is a discrete random variable described by   
\begin{equation}\label{eq:antenna} 
\alpha_i = 
\begin{cases}
\alpha_\text{L} & \text{w.p.} \hspace{2mm} \ell_p(x) \\
\alpha_\text{N} & \text{w.p.} \hspace{2mm} \ell_p(x) \\
\end{cases},
\end{equation}
where $\alpha_\text{L}$ and $\alpha_\text{N}$ are the LOS and NLOS path-loss exponents and $\ell_p(x)$ is the probability a link of length $x$ is LOS. Fig. \ref{fig:network} shows an example realization of the \emph{ad hoc} network. The density and mean building size are modeled to match The University of Texas at Austin \cite{Bai2013}. We ignore correlations between blockages, as in \cite{Bai2013}; the blockage on each link is determined independently. While the correlations might affect the tail behavior of the SINR distribution \cite{BaccelliZhang2015}, it was shown that the difference in the practical operating SINR range is small when ignoring the correlation \cite{Bai2012}. Moreover, simulations that use real geographical data \cite{Singh2015,WeiLu} match analytical expressions ignoring blockage correlation. %By using stochastic geometry, the analysis averages the spatial positions of the users such that second-order effects like position correlation are minimized.

\subsection{SINR}
To help with the analytical tractability, we model the fading as a Nakagami random variable with parameter $N_\text{h}$. Consequently, the received signal power, $h$, can be modeled as a gamma random variable, $h\sim\Gamma(N_\text{h},1/N_\text{h})$. As $N_\text{h} \rightarrow \infty$, the fading becomes a deterministic value centered on the mean, whereas $N_\text{h} = 1$ corresponds to Rayleigh fading.
  
The SINR is the basis of the performance metrics in this paper. $P_\text{t}$ is the transmit power of each dipole, $G_0$ is the antenna gain corresponding to both main beams aligned, $h_0$ is the fading power at the dipole of interest, $A$ is the path-loss intercept, $r$ is the fixed dipole link length, $\alpha_0$ is the path-loss exponent, and $N_0$ is the noise power. The interference term for each interfering dipole transmitter is indexed with $i$: $d_i$ is used to represent the distance from the interferer to transmitter of interest, $h_i$ is each interference fading power distributed IID according to a gamma distribution, and $M_i$ is the discrete random antenna gain distributed IID according to \eqref{eq:ant_gain}. The SINR is defined as \cite{baccelli_app}
\begin{equation}
\text{SINR} = \frac{P_\text{t} G_0 h_0 A r^{-\alpha_0}}{N_0 + \sum_{i \in \Phi} P_\text{t} M_i h_i A d_i^{-\alpha_i}}.\label{eq:SINR}
\end{equation}

\begin{figure}
  \centering
  \begin{subfigure}{.5\textwidth}
  \centering
  \includegraphics[width=.65\linewidth]{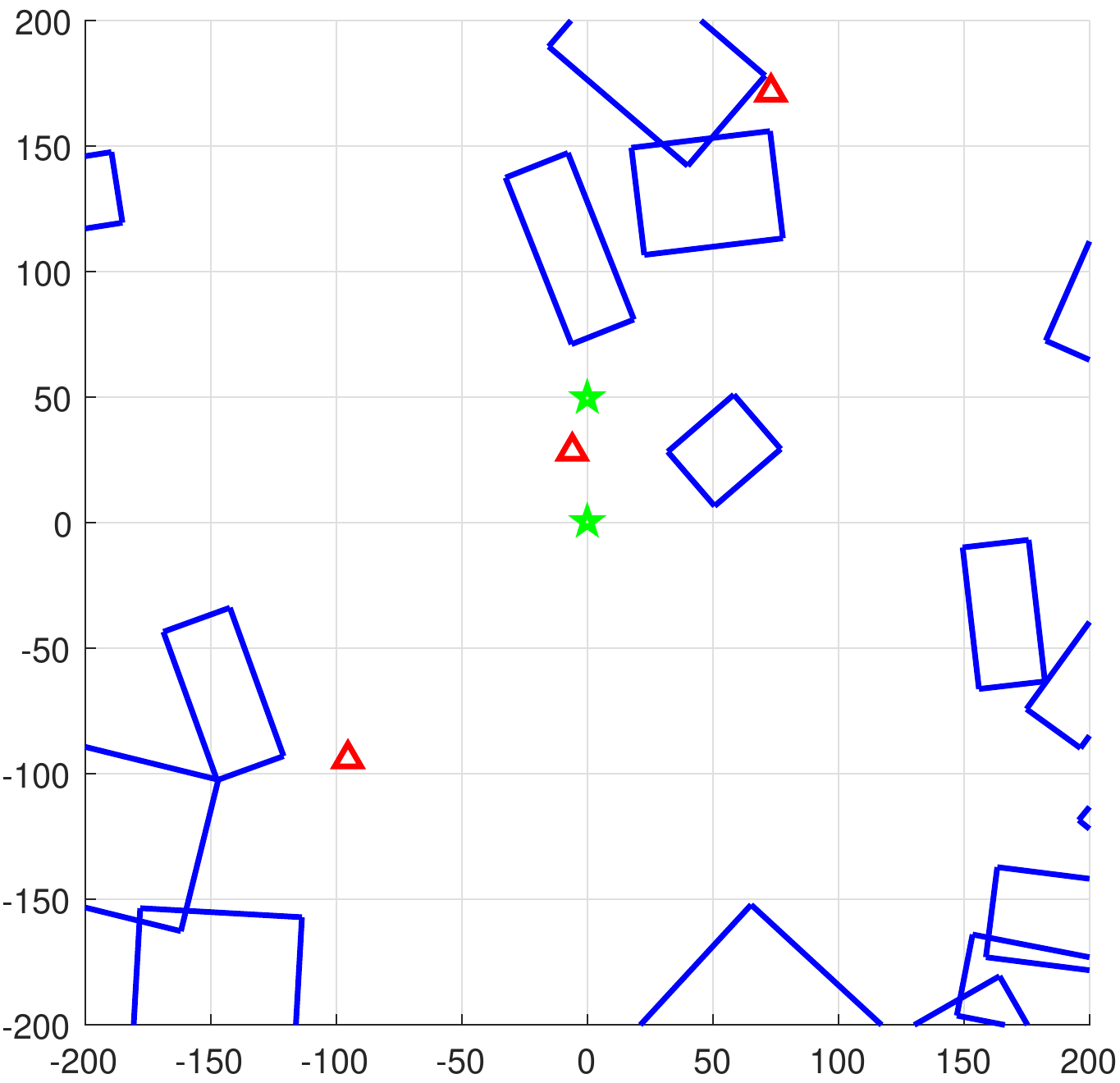}
 \caption{}
  \label{fig:sparse_network}
  \end{subfigure}
  \begin{subfigure}{.5\textwidth}
  \centering
  \includegraphics[width=.65\linewidth]{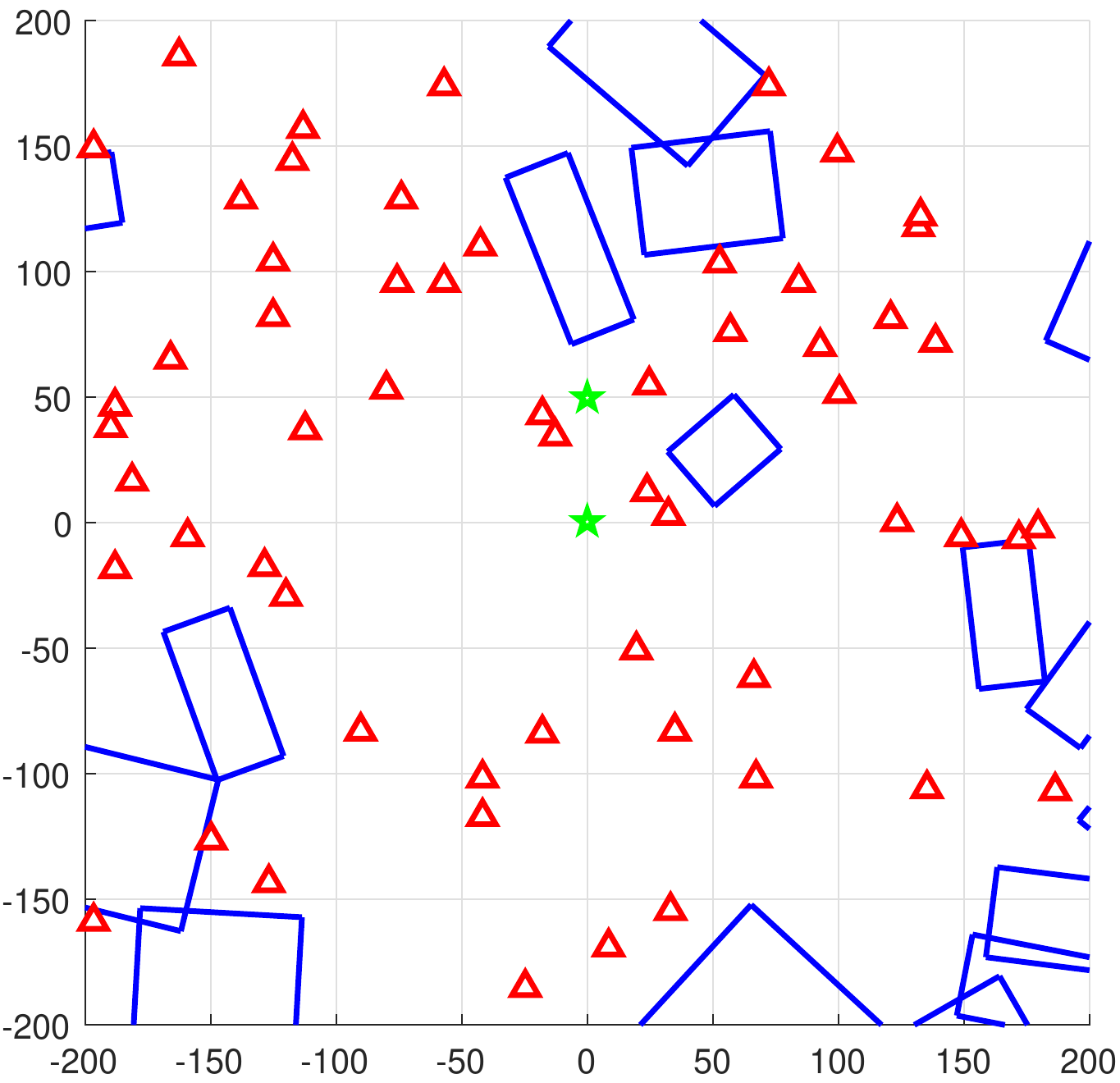}
  \caption{}
  \label{fig:dense_network}
  \end{subfigure}
	\caption{Example realizations of the random network with blockage. The blue rectangles are random boolean buildings which attenuate the signal. The red triangles are the Poisson point process of interferers. The green star represents the \emph{typical node}. The user densities are what we call \emph{sparse} (a) and \emph{dense} (b) when discussing the results.}
	\label{fig:network}
  \end{figure}

%When discussing the results and system performance in the subsequent sections, we will interpret the network density, $\lambda$, as the average \emph{neighbor} distance which allows some intuition on the spacing of nodes defined as $d_{\text{n}} = \frac{1}{2\sqrt{\lambda}}$ in $\bbR^2$. Additionally, we define the expected number of LOS interfering nodes as $\rho = \bbE[\# \text{LOS}] = 2\pi\lambda/\beta^2$, which follows as a direct result of Campbell's Theorem \cite{Bai2014}. 

\section{One-Way Ad Hoc Communication}\label{sec:oneway}
In this section, we derive the SINR distribution for one-way transmission in the \emph{ad hoc} network described in Section \ref{sec:sys}. We first characterize the \emph{overall} SINR complimentary cumulative distribution function (CCDF) by analyzing the network when the desired link is either LOS and NLOS. Next, we define the \emph{protocol-gain} by limiting communication to LOS links and argue why this is a useful concept. We quantify the effect of random receiver distance. We show that neglecting noise and NLOS interference does not change the SINR distribution, suggesting that mmWave \emph{ad hoc} networks are \emph{LOS interference limited} in dense networks. This is reinforced by the derivation of the INR cumulative distribution function CDF. Lastly, the performance metrics, transmission capacity and area spectral efficiency, are computed using a bound of the SINR CCDF.

\subsection{SINR Distribution}\label{sec:sinroneway}
We define the CCDF of the SINR as
\begin{equation}\label{eq:ccdf}
P_\text{c}(T) = \bbP[\text{SINR}\geq T],
\end{equation} 
where $T$ is target SINR. In other work, \eqref{eq:ccdf} is referred to as the \emph{coverage probability} \cite{baccelli_app,WebAnd10,Hunter08}. We can use the law of total probability to expand the SINR CCDF as \cite{Bai2014}
\begin{equation}
P_\text{c}(T) = P_\text{c}^{\text{L}}(T)\bbP[\text{LOS}]+ P_\text{c}^{\text{N}}(T)\bbP[\text{NLOS}]\label{eq:full_eq},
\end{equation}
where $P_\text{c}^{\text{L}}$ and $P_\text{c}^{\text{N}}$ are the conditional CCDFs on the event that the main link is LOS and NLOS, respectively. The SINR CCDF conditioned on the link being LOS is \cite{Bai2014}
\begin{equation}\label{eq:cond_ccdf}
P^\text{L}_\text{c}(T) = \bbP[\mathrm{SINR} \geq T | \text{LOS}]. 
\end{equation}

Going forward, for brevity, we will drop the conditional notation when using $P^\text{L}_c$. Using (\ref{eq:SINR}), 
\begin{align}\allowdisplaybreaks
P^\text{L}_\text{c}(T) &= \bbP\bigg[\frac{P_\text{t} G_0 h_0 A r^{-\alpha_\text{L}}}{N_0 + \sum_{i \in \Phi} P_\text{t} M_i h_i A d_i^{-\alpha_i}} \geq T\bigg] \\
%&=& \bbP\bigg[\frac{{P_\text{t} G_0 h_0 }}{r^{\alpha_\text{L}}} > T({N + \sum_{i \in \Phi} \frac{P_\text{t} M_i h_i }{d_i^{\alpha_i}})}\bigg]\\
&= \bbP\bigg[{h_0} \geq \frac{T r^{\alpha_\text{L}}}{P_\text{t} G_0 A}\bigg({N_0 + \sum_{i \in \Phi} \frac{P_\text{t} M_i h_i A }{d_i^{\alpha_i}}\bigg)}\bigg]\\
&= \bbP\bigg[{h_0} \geq \frac{T r^{\alpha_\text{L}}}{P_\text{t} G_0 A}({N_0+ I_\Phi)}\bigg] \\
&= 1 - \bbP\bigg[{h_0} < \frac{T r^{\alpha_\text{L}}}{P_\text{t} G_0 A}({N_0+ I_\Phi)}\bigg]\\
\begin{split} 
&= 1 - \int_{0}^{\infty}\bbP\bigg[{h_0} < \frac{T r^{\alpha_\text{L}}}{P_\text{t} G_0 A}({N_0+ x)}|I_\Phi = x\bigg]\times \\
&\hspace{5mm}p_\Phi(x)dx\label{eq:SINRprob}, 
\end{split}
\end{align}
where $I_\Phi$ is the aggregate interference due to the PPP and $p_\Phi$ is the probability distribution function of the PPP. We introduce the following Lemma to aid the analysis.
\begin{lemma}
	The cumulative distribution function of a normalized gamma random variable with integer parameter $k$, $y\sim\Gamma(k,1/k)$, can be tightly lower bounded as
	\begin{equation*}
	\left[1-e^{-az}\right]^k < \bbP\left[y < z\right]
	\end{equation*}
	with $a = k(k!)^{-1/k}$.\\
	\begin{proof}
		See Appendix A.
	\end{proof}
\end{lemma}
Now we can bound (\ref{eq:SINRprob}) as
\begin{align}
P^\text{L}_\text{c}(T) &< 1 - \int_{0}^{\infty}\bigg[\bigg(1-e^{-a  \frac{T r^{\alpha_\text{L}}}{P_\text{t} G_0 A}({N_0+ x)}}\bigg)^{N_\text{h}}\bigg]p_\Phi(x)dx\\
 &= 1 -\bbE_\Phi \bigg[\bigg(1-e^{-a  \frac{T r^{\alpha_\text{L}}}{P_\text{t} G_0 A}({N_0+ I_\Phi)}}\bigg)^{N_\text{h}}\bigg]\\
&= \sum_{n = 1}^{N_\text{h}} {N_\text{h} \choose n} (-1)^{n+1} \bbE_\Phi \big[e^{-a n  \frac{T r^{\alpha_\text{L}}}{P_\text{t} G_0 A}({N_0+ I_\Phi)}}\big] \label{eq:SINR_approx},
\end{align}
where (\ref{eq:SINR_approx}) is from the Binomial Theorem \cite{Bai2014}.

Because the correlation between each random blockage is ignored, each point in the building blockage PPP is independent which permits the use of the thinning theorem from stochastic geometry \cite{baccelli_theory}. We further thin $\Phi$ based on the random antenna gain. Essentially, we can now view the interference as 6 independent PPPs such that
\begin{equation}
I_\Phi = I_{\Phi_{\text{LOS}}}^{\text{GG}}+I_{\Phi_{\text{LOS}}}^{\text{Gg}}+I_{\Phi_{\text{LOS}}}^{\text{gg}}+I_{\Phi_{\text{NLOS}}}^{\text{GG}}+I_{\Phi_{\text{NLOS}}}^{\text{Gg}}+I_{\Phi_{\text{NLOS}}}^{\text{gg}},\label{eq:subPPP}
\end{equation}
with the superscripts representing the discrete random antenna gain defined in (\ref{eq:ant_gain}) and each interfering node either a LOS transmitter or NLOS transmitter. We can distribute the expectation in (\ref{eq:SINR_approx}) as
\begin{equation}
P^\text{L}_\text{c} < \sum_{n=1}^{N_\text{h}} (-1)^{n+1}{N_\text{h} \choose n} e^{-nK_\text{L}TN_0} \prod_{i} \prod_{j}\bbE_{I_{\Phi_{j}^i}}\big[e^{-nK_\text{L}TI_{\Phi^i_{j}}}\big] \label{eq:PCL}
\end{equation}
with $i \in \{\text{GG},\text{Gg},\text{gg}\}$, $j \in \{\text{LOS},\text{NLOS}\}$, and $K_\text{L} = \frac{a r^{\alpha_\text{L}}}{P_\text{t} G_0 A}$. In \eqref{eq:PCL}, $i$ and $j$ index each interference sub-PPP. In essence, each expectation is a the Laplace transform of the associated sub-PPP, and each of these Laplace transforms are multiplied together. 

Using stochastic geometry, we can analytically represent the first Laplace expectation term as
\begin{align}
\bbE&\bigg[e^{-n K_\text{L} T I_{\Phi_{\text{LOS}}}^\text{GG}}\bigg] = \\ \notag
 &e^{-2\pi \lambda p_\text{GG} \int_{0}^{\infty}\left(1 - \bbE_h\big[e^{-\frac{n K_\text{L} T P_\text{t} A GG h}{x^{\alpha_\text{L}} }}\big]\right)\ell_p(x)x \text{d}x },\label{eq:SINRsingleLap}
\end{align}
where $p_\text{GG}$ and $\ell_p(x)$, capture the thinning of the PPP for the first sub-PPP in \eqref{eq:subPPP}. Notice that $\bbE_h[e^{\eta h}]$ corresponds to the moment-generating function (MGF) of the random variable $h$ (e.g. gamma). A similar approach was taken in \cite{Bai2014} for the analysis of mmWave cellular networks. The final Laplace transform of the PPP is given as
\begin{equation}
\cL_{I_{\Phi_{\text{LOS}}}^{\text{GG}}} =  e^{-2\pi \lambda  p_\text{GG}\int_{0}^{\infty}\big(1 -1 / (1+\frac{n Q_\text{L} T }{x^{\alpha_\text{L}} N_\text{h}})^{N_\text{h}}\big)  \ell_p(x) x \text{d}x }.\label{eq:SINRlaptrans}
\end{equation}
with $Q_\text{L} = K_\text{L} P_\text{t} GG A = \frac{a r^{\alpha_\text{L}}GG}{G_0}$. Each other Laplace transform is computed similarly, but noting that $p_\text{GG}$, $\ell_p(x)$, and $x^{\alpha_\text{L}}$ will change depending on the antenna gain of the sub-PPP and if the sub-PPP is LOS or NLOS. We can summarize the results in the following theorem
 
 \setcounter{theorem}{0}
 \begin{theorem}
 	The SINR distribution of an outdoor mmWave \emph{ad hoc} network can be tightly upper bounded by
 	\begin{align}
 		\begin{split}
 	& P_\text{c}(T) <  \sum_{n = 1}^{N_\text{h}} {N_\text{h} \choose n} (-1)^{n+1} e^{-nK_\text{L}TN_0}  e^{-2\pi \lambda (\kappa_\text{L} + \kappa_\text{N})}\ell_p(r)\\ &+ \sum_{n = 1}^{N_\text{h}} {N_\text{h} \choose n} (-1)^{n+1} e^{-nK_\text{N}TN_0}  e^{-2\pi \lambda (\xi_\text{L}  + \xi_\text{N}  )}\bigg(1-\ell_p(r)\bigg)
 	\end{split}
\end{align}
where
\begin{align}
\kappa_\text{L} &= \sum_{i} p_i\int_{0}^{\infty}\bigg[1 -1 / \bigg(1+\frac{n Q_\text{L} T }{x^{\alpha_\text{L}} N_\text{h}}\bigg)^{N_\text{h}}\bigg]  \ell_p(x) x \text{d}x \\
\kappa_\text{N} &= \sum_{i} p_i\int_{0}^{\infty}\bigg[1 -1 / \bigg(1+\frac{n Q_\text{L} T }{x^{\alpha_\text{L}} N_\text{h}}\bigg)^{N_\text{h}}\bigg] \big(1- \ell_p(x)\big) x \text{d}x \\
\xi_\text{L}  &= \sum_{i} p_i\int_{0}^{\infty}\bigg[1 -1 / \bigg(1+\frac{n Q_\text{N} T }{x^{\alpha_\text{N}} N_\text{h}}\bigg)^{N_\text{h}}\bigg]  \ell_p(x) x \text{d}x\\
\xi_\text{N}   &= \sum_{i} p_i\int_{0}^{\infty}\bigg[1 -1 / \bigg(1+\frac{n Q_\text{N} T }{x^{\alpha_\text{N}} N_\text{h}}\bigg)^{N_\text{h}}\bigg] \big(1- \ell_p(x)\big) x \text{d}x 
\end{align}
with $K_\text{L} = \frac{a r^{\alpha_\text{L}}}{P_\text{t} G_0 A}$, $K_\text{N} = \frac{a r^{\alpha_\text{N}}}{P_\text{t} G_0 A}$, $i\in\{\text{GG,Gg,gg}\}$, $Q_\text{L} =\frac{a r^{\alpha_\text{L}}M_i}{G_0}$, and $Q_\text{N} = \frac{a r^{\alpha_\text{N}}M_i}{G_0}$.
 \end{theorem}
 \begin{proof}
 	Substituting each Laplace transform (\ref{eq:SINRlaptrans}) into (\ref{eq:PCL}) for the conditional $P_\text{c}^\text{L}$ yields the first summation in Theorem 1. The same process is done for the Laplace transforms corresponding to $P_\text{c}^\text{N}$. These summations are then multiplied by $\bbP[\text{LOS}]$ and $\bbP[\text{NLOS}]$, respectively, to give the full CCDF of (\ref{eq:full_eq}).
 \end{proof}
 
 In Theorem 1, $\kappa_\text{L}$ and $\kappa_\text{N}$ correspond to the LOS and NLOS interference, respectively, when the desired signal is LOS while $\xi_\text{L}  $ and $\xi_\text{N}  $ correspond to the LOS and NLOS interference, respectively, when the desired signal is NLOS. While Theorem 1 may appear unwieldy, the decomposition of the terms illustrates the insight that can be gained from the Theorem. In the first summation, there are exponential terms that correspond to noise, LOS interference (i.e. $\kappa_\text{L}$), and NLOS interference (i.e. $\kappa_\text{N}$). Further, both $\kappa_\text{L}$ and $\kappa_\text{N}$ (and similarly $\xi_\text{L}  $ and $\xi_\text{N}  $) can be decomposed based on each antenna gain. It is possible to compare relative contributions to the total SINR CCDF. For example, by computing $\kappa_\text{N}$, we were able to see that $\kappa_\text{L} \gg \kappa_\text{N}$ for many different system parameters of interest. Therefore, $e^{-2\pi\lambda(\kappa_\text{L}+\kappa_\text{N})} \approx e^{-2\pi\lambda \kappa_\text{L}}$ which means NLOS interference has relatively no effect on the SINR distribution. We use this insight in Section \ref{sec:LOSint} to conclude that mmWave \emph{ad hoc} networks are \emph{LOS interference limited}.
 
 \subsection{Validation of the Model}
 
 Before proceeding, we verify the tightness of the bound in Theorem 1. Table \ref{params} shows the values used throughout the section.  
   \begin{table}
   \centering
       \begin{tabular}{ | l | l | l | l |}
       \hline
       \textbf{Parameter}& \textbf{Value} \\ \hline
       $\lambda$ & $5\times10^{-5}$, $5\times10^{-4}$ ($m^{-2}$)\\ \hline
       $r$ & 25, 50, 75 ($m$)\\ \hline
       $\beta$, $\alpha_{\text{LOS}}$, $\alpha_{\text{NLOS}}$& 0.008, 2, 4 \\ \hline
       $N_0$& -117 dB \\ \hline    
       $h_i$, $N_\text{h}$ &  Gamma, $7$\\ \hline
       $\theta$, $M$, $m$ & $\frac{\pi}{6}$, $10$, $0.1$  \\ \hline
       $P_\text{t}$ & 1W (30dBm) \\ \hline      
       \end{tabular}
       \caption{Parameters of results.}
               \label{params}
   \end{table}   
The parameters of (\ref{eq:full_eq}) are simulated through Monte Carlo, while Theorem 1 is used for the analytical model. For the simulation, a PPP was generated over an area of $4 \text{km}^2$. The thermal noise power of 500MHz bandwidth at room temperature is $-117$dB. We used $N_\text{h}=3$ when computing the analytical expressions. We chose $N_\text{h} = 3$ because measurement campaigns have shown that small-scale fading is more deterministic at mmWave \cite{Rappaport2015}.  In the measurements of \cite{Samimi2013,RapSun13}, small-scale fading is not very significant. Because of the directional antennas and sparse channel characteristics, the uniform scattering assumption for Rayleigh fading is not valid at mmWave frequencies. We chose a $30^\circ$ beamwidth. Additionally, 10dB gain corresponds to the theoretical gain of a 10 element uniform linear array unit gain antennas.
  
Fig. \ref{fig:sub_lam_1} compares the analytical SINR distribution with the empirical given a $\lambda = 5\times 10^{-5} \text{m}^{-2}$ or an average of $50$ $\text{users}/\text{km}^2$. This can be attributed to the directional antennas limiting the interference seen by the typical node. The analytical expression in Theorem 1 of the mmWave \emph{ad hoc} network matches extremely well to the simulations. For all three link lengths, the SINR of the users is greater than $0$dB a majority of the time. 
 
 \begin{figure}
   \centering
   \begin{subfigure}{.5\textwidth}
   \centering
   \includegraphics[width=.65\linewidth]{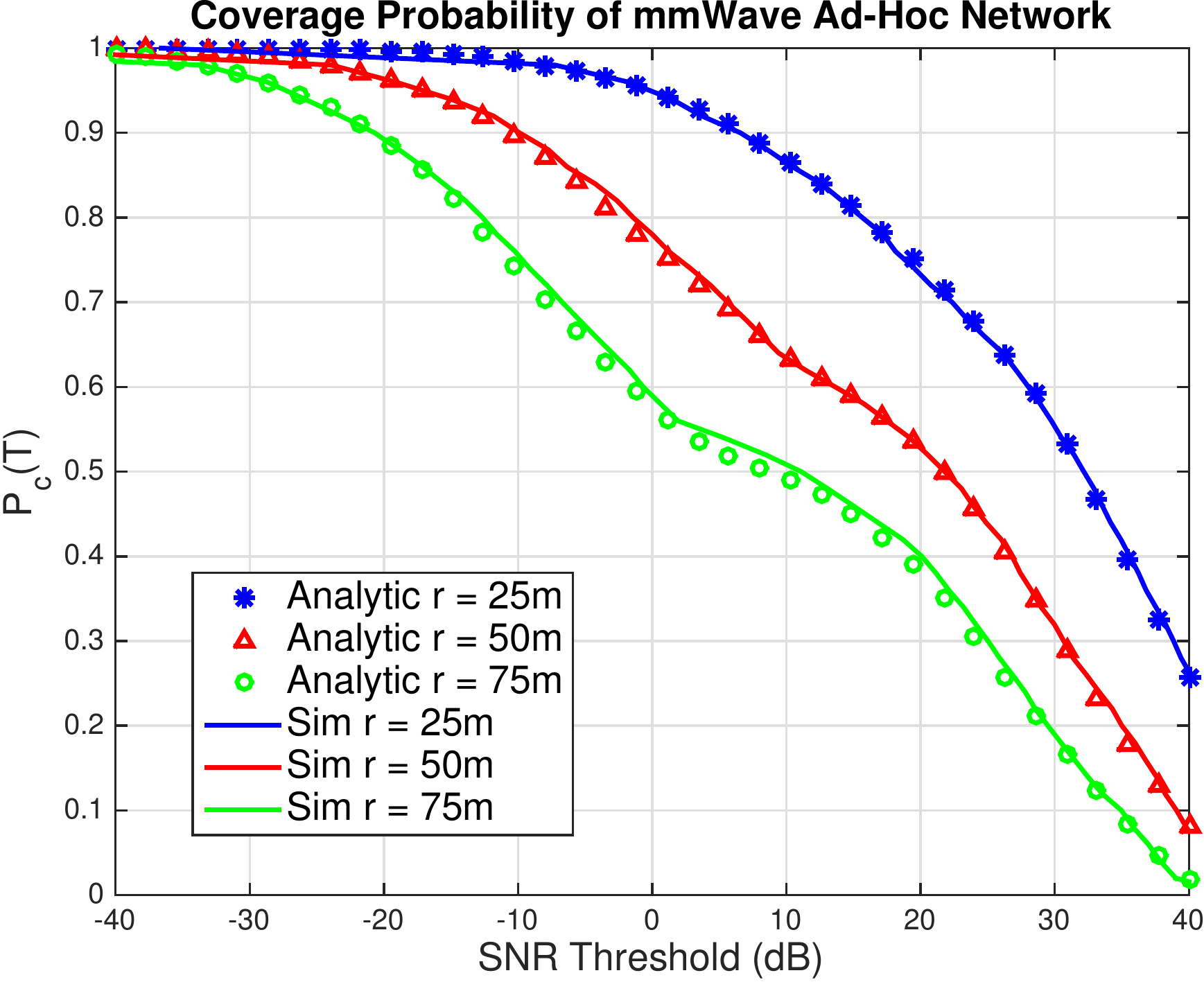}
  \caption{}
   \label{fig:sub_lam_1}
   \end{subfigure}
   \begin{subfigure}{.5\textwidth}
   \centering
   \includegraphics[width=.65\linewidth]{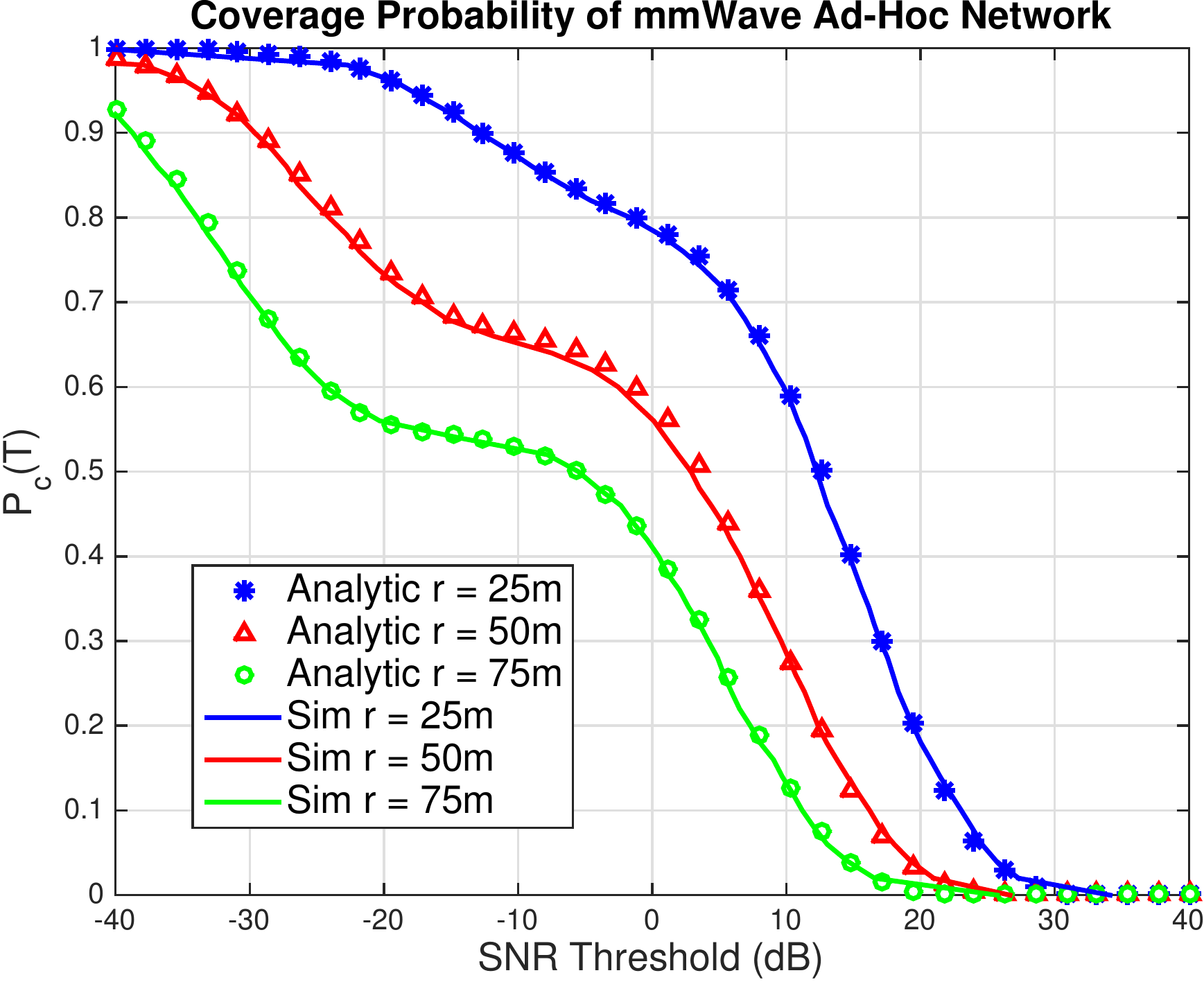}
   \caption{}
   \label{fig:sub_lam_2}
   \end{subfigure}
\caption{The SINR distribution of mmWave \emph{ad hoc} networks with $\lambda = 5\times10^{-5}$ (a) and $\lambda = 5\times10^{-4}$ (b). }
  	\label{fig:overall_coverage}
   \end{figure}

 Fig. \ref{fig:sub_lam_2} compares the SINR distribution results for a much denser network, $\lambda = 5\times 10^{-4}\text{m}^{-2}$ which corresponds to an average of $500$ $\text{users}/\text{km}^2$. Again, Theorem 1 matches the simulation well. For the larger link distances, we see bi-modal behavior of the CCDF with the plateaus around $-10$dB.
    
  \subsection{LOS Protocol-Gain}
In this section, we define and discuss the \emph{LOS protocol-gain}. We can view $P_\text{c}(T)$ as a mixture of $P_\text{c}^\text{L}(T)$ and $P_\text{c}^\text{N}(T)$. In Fig. \ref{fig:sub_lam_2}, the interference causes most of the density of $P_\text{c}^\text{N}(T)$ to shift to very low SINR. The plateaus in the CCDF of Fig. \ref{fig:sub_lam_2} illustrate this separation. Unless the SINR threshold is very low (e.g. below -20dB), these links will not be able to communicate without LOS communication. This motivates the need for a protocol to ensure LOS communication (e.g. using a LOS relay to multi-hop around a building). If LOS communication is assumed, the SINR distribution in the LOS regime will be equal to $P_\text{c}^\text{L}(T)$ (i.e. set $\bbP[\text{LOS}] = 1$). With many users nearby, the network will have multiple users that could potentially be a LOS receiver.

 \begin{figure}
   \centering
   \begin{subfigure}{.5\textwidth}
   \centering
   \includegraphics[width=.65\linewidth]{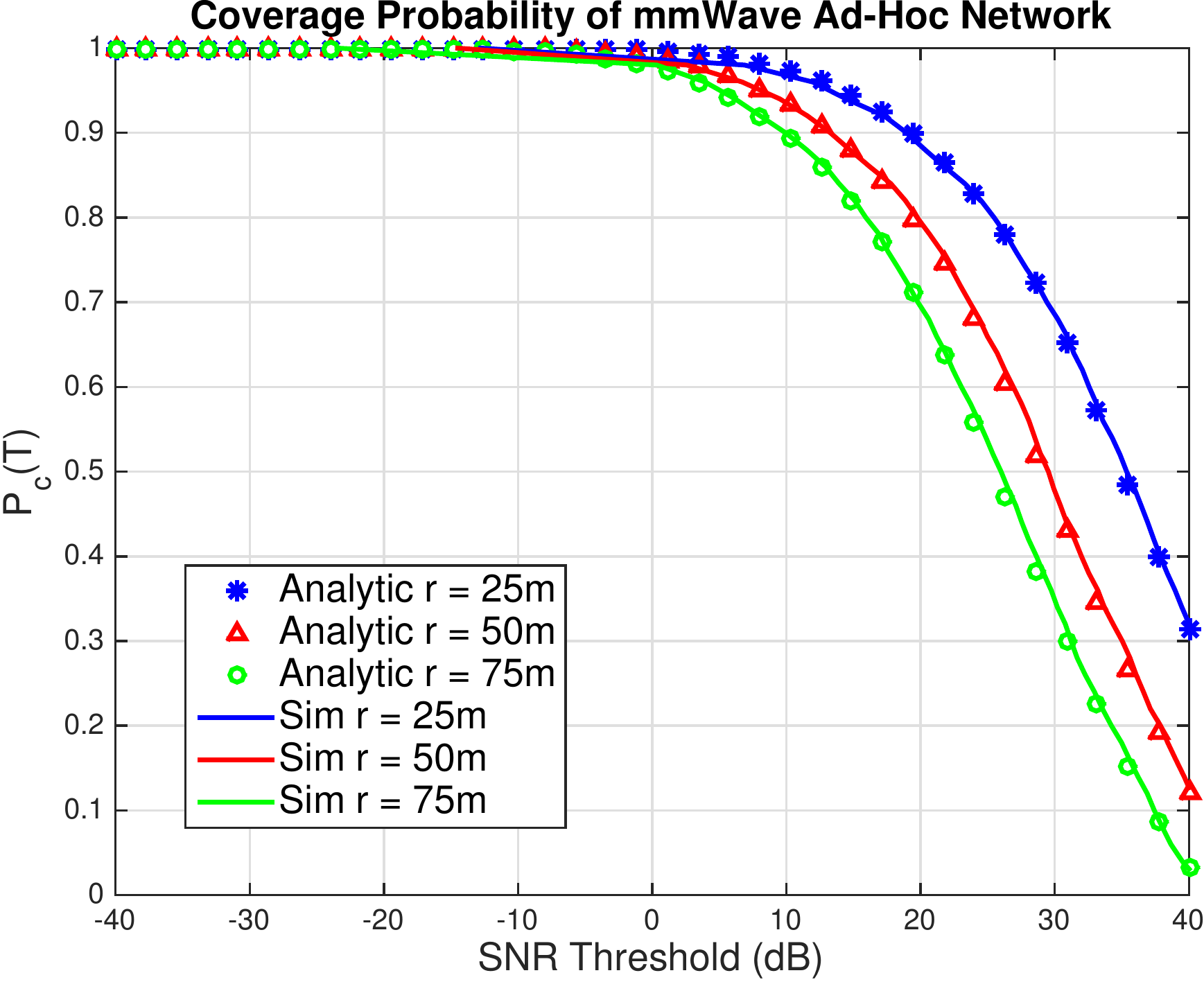}
  \caption{}
   \label{fig:sub_lam1_los}
   \end{subfigure}
   \begin{subfigure}{.5\textwidth}
   \centering
   \includegraphics[width=.65\linewidth]{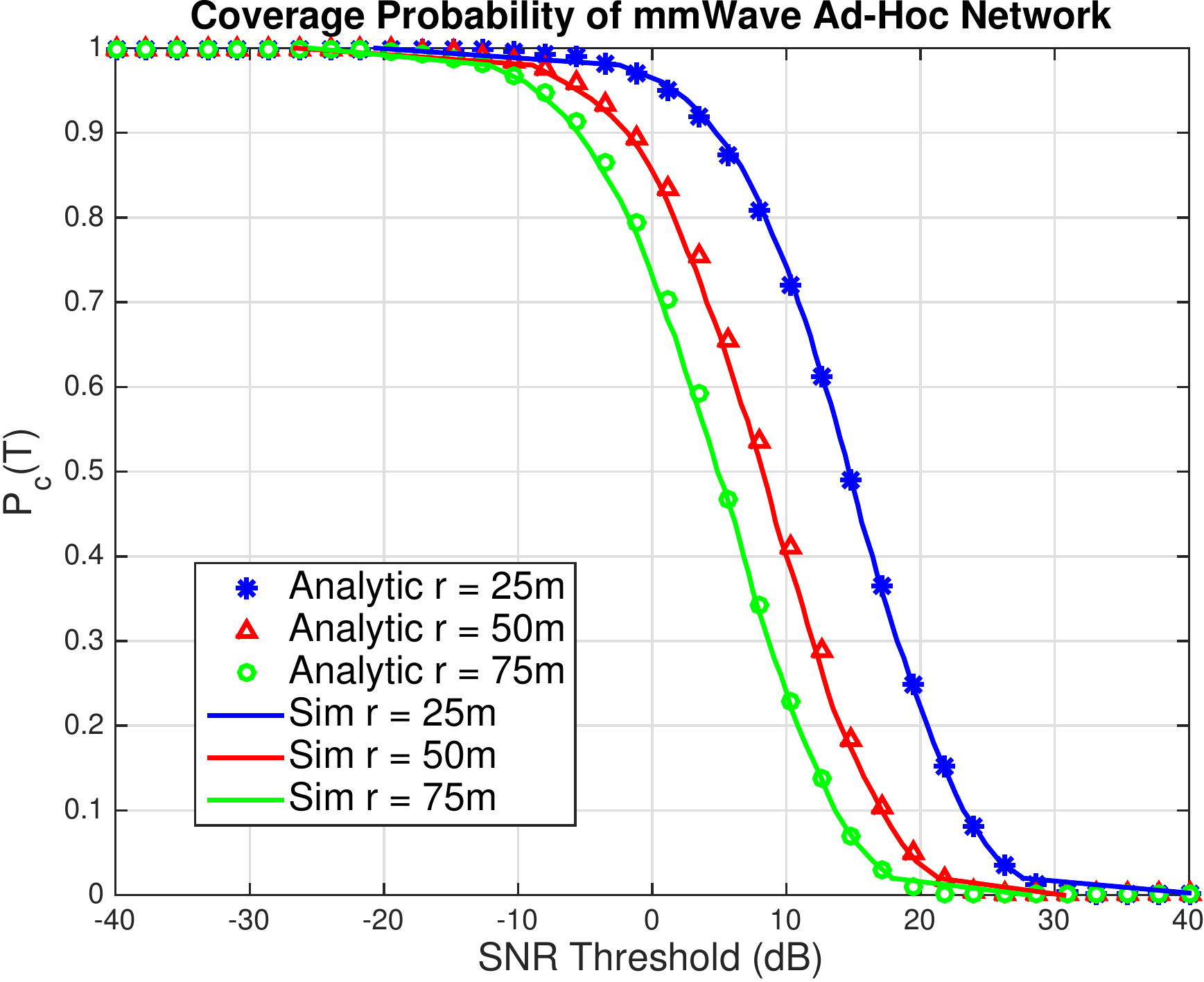}
   \caption{}
   \label{fig:sub_lam2_los}
   \end{subfigure}
	\caption{The SINR distribution of mmWave \emph{ad hoc} networks with $\lambda = 5\times10^{-5}$ (a) and $\lambda = 5\times10^{-4}$ (b). If the desired link is LOS, significant improvement to the SINR distribution is realized. We term this the \emph{LOS protocol-gain}.}
 	\label{fig:los_coverage}
   \end{figure} 
   
 Fig. \ref{fig:los_coverage} shows the SINR distribution of a mmWave \emph{ad hoc} network if the desired link is LOS. The improvement is quite large. The 90\% coverage point in Fig. \ref{fig:sub_lam1_los} is improved by 10dB for 25m, 20dB for 50m, and 30dB for 75m, compared to the same network in \ref{fig:sub_lam_1}. The improvement in Fig. \ref{fig:sub_lam2_los} is even more drastic. For the 25m link, 20dB improvement is seen. This knowledge should influence MAC design, which is why we call it \emph{protocol-gain}.

%In the event that no LOS receivers are available, which is still possible in a dense network, using a spread-spectrum technique is advised, as is done in the 802.11ad standard \cite{80211ad}. Because of the large bandwidth (e.g. 2GHz for 802.11ad), bandwidth can be sacrificed in order to increase received SINR which will improve the distribution. For example, if orthogonal codes are used with spreading factor $K$, the maximum achievable rate with spread spectrum is 
%\begin{equation}
%R^{\text{ss}} = \underset{K}{\operatorname{max}}  \hspace{2mm}\frac{B}{K}\text{log}_2(1+K T) P_\text{c}(T/K),
%\end{equation}
%where $B$ is the system bandwidth. We leave the exploration of this trade-off to future work \cite{Weber2005}.
 
\subsection{Distributions of $r$}
 
One of the limitations of the dipole model is the fixed length of the communication link. This model is used for its analytical tractability but is not a realistic expectation. In a D2D gaming scenario, for example, the distance between the receiver and transmitter will vary as the two users walk around. To quantify this, we can integrate Theorem 1 against a receiver location density function. The SINR distribution accounting for different receiver geometries is
\begin{equation}
P_\text{c}^r(T) = \int_{S} P_\text{c}(r,T) f_R(r) dr
\end{equation}
 where $S$ is the support of the location density distribution and $f_R$ is the density and $P_\text{c}(r,T)$ is Theorem 1, but we allow varying receiver distances. We compare two different distributions against the fixed dipole assumption. 
    
     \begin{figure}
     	\centering
     	\includegraphics[width=0.65\columnwidth]{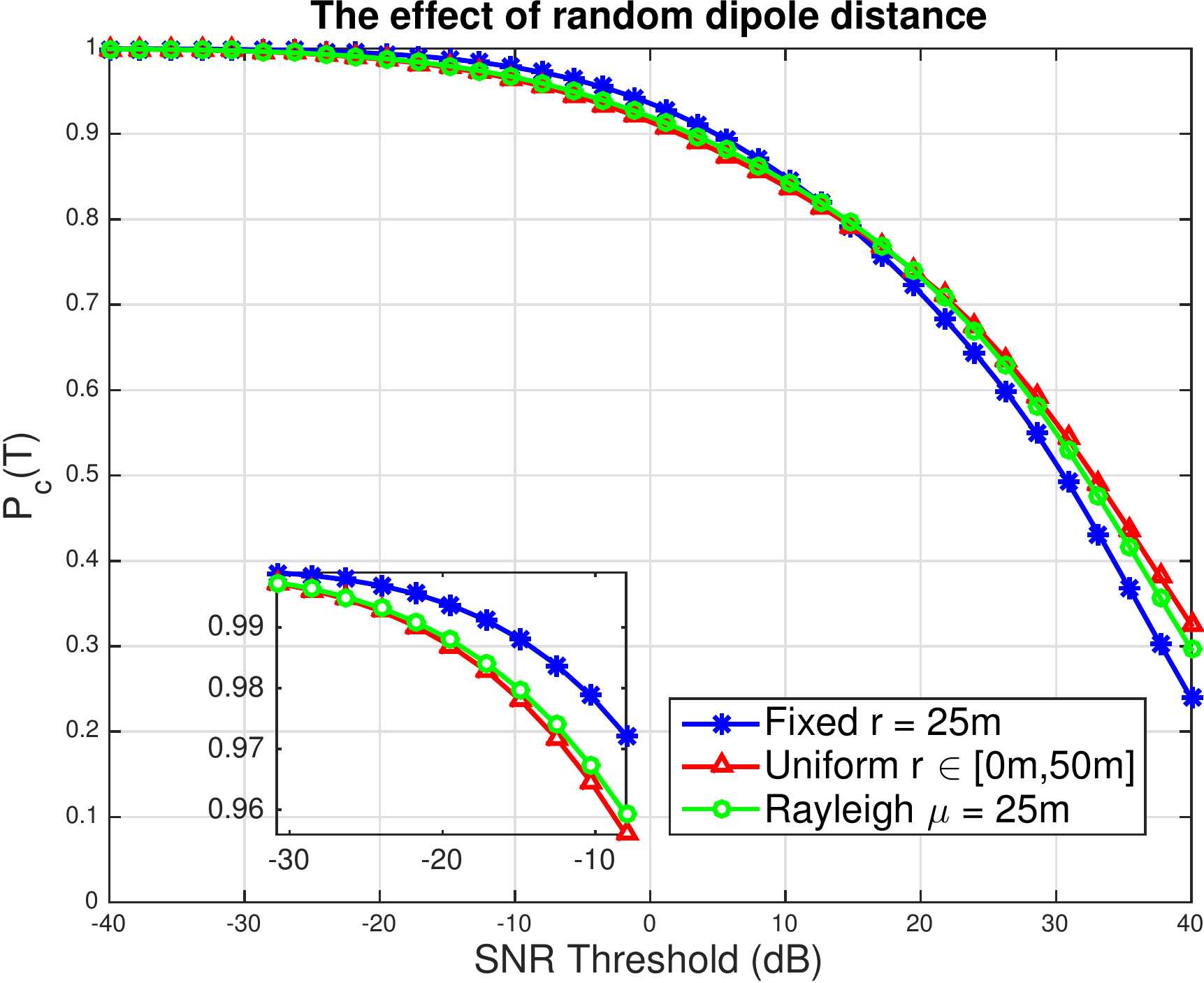}
     	\caption{The effect of receiver distribution is quantified for the overall (LOS/NLOS) SINR distribution (a) and LOS-only SINR distribution (b). Each link, on average, is 25m.}
     	\label{fig:var_r}
     \end{figure}
 
  As shown in \figref{fig:var_r}, we use two receiver geometries to compare against, the uniform and Rayleigh \cite{Lin2014}. For larger SINR thresholds, including a random receiver distance improves performance. This is due to the positive effect of having the receiver closer some of the time. As shown in \figref{fig:overall_coverage}, communication when NLOS generally has poor SINR. The random shorter link means LOS communication is more likely. Conversely, the random receiver locations hurt performance for lower SINR thresholds. If assuming random receiver locations, both distributions give similar results despite the Rayleigh distribution having unbounded support. Surprisingly, the results indicate that simply knowing the mean of the distribution captures much of the SINR distribution. 
 
 \subsection{LOS Interference Limited Networks}\label{sec:LOSint}
 Interference is a key design limitation for \emph{ad hoc} networks. Cellular network analysis has shown that mmWave cellular networks can be modeled as noise-limited with inter-site-distances of 200m \cite{Singh2015,Kulkarni2014,Bai2014,RapSun13}. This network topology, however, is different from an \emph{ad hoc} network as cellular users associate with a fixed base station. We now characterize the transition from noise-limited to interference-limited operation as a function of user density, building density, antenna pattern, and transmission distance. We achieve this by using the interference-to-noise ratio (INR) cumulative distribution function (CDF) 
 \begin{equation}
 P^\text{NL}(T) = \bbP[\mathrm{INR} \leq T ]. \label{eq:inr}
 \end{equation}
 We leave the threshold value up to system designers to determine what value of $T$ is appropriate for defining noise limited. A natural choice may be 1 (0dB) or 10 (10dB). The INR CDF can be written as
 \begin{eqnarray}\allowdisplaybreaks
 P^\text{NL}(T) &=& \bbP\bigg[\frac{\sum_{\text{i} \in \Phi} P_\text{t} M_\text{i} h_\text{i} A r_\text{i}^{-\alpha_\text{i}}}{N_\text{0}} \leq T\bigg] \\
 &=& \bbP\bigg[1 \geq \frac{\sum_{\text{i} \in \Phi} P_\text{t} M_\text{i} h_\text{i} A r_\text{i}^{-\alpha_\text{i}}}{T N_\text{0}} \bigg] \\
 &=& \bbP\bigg[1 \geq \frac{I_\Phi}{T N_\text{0}}\bigg]\\
 &=& 1 - \bbP\bigg[1 < \frac{I_\Phi}{T N_\text{0}}\bigg].\label{eq:prob}
 \end{eqnarray}
To analytically evaluate $\bbP\left[1 < \frac{I_\Phi}{T N_\text{0}}\right]$, we replace $1$ with a random variable, $C$, with low variance. We let $C\sim\Gamma(N_\text{C},1/N_\text{C})$. If we examine the probability density function (PDF) of $C$,
\begin{equation}
f_C(x) = \frac{N_\text{C}^{N_\text{C}} x^{N_\text{C}-1}e^{-N_\text{C}x}}{\Gamma(N_\text{C})},
\end{equation}
the $\lim\limits_{N_\text{C}\rightarrow \infty} f_C(x) = \delta(x-1)$. Further, we leverage Lemma 1 again. The INR distribution can then be bounded as
 \begin{align}
 P^{\text{NL}} &= 1 - \bbP\bigg[C < \frac{I_\Phi}{T N_\text{0}}\bigg]\\
 &< 1 -\bbE_\Phi \bigg[\big(1-e^{-a \frac{I_\Phi}{T N_0}}\big)^{N_\text{C}}\bigg]\label{eq:pdfineq}\\
 &= \sum_{n = 1}^{N_\text{C}} {N_\text{C} \choose n} (-1)^{n+1} \bbE_\Phi \big[e^{-a n \frac{I_\Phi}{T N_0}}\big] \label{eq:binom},
 \end{align}
 where (\ref{eq:pdfineq}) is from the law of total probability and gamma CDF approximation while (\ref{eq:binom}) is from the Binomial Theorem. The transmitters, again, are six independent PPPs as explained in (\ref{eq:subPPP}). Because each sub-process is independent, we re-write (\ref{eq:binom}) as a product of expectations. The analytic expression of the first Laplace expectation term is

 \begin{equation}
 \bbE\bigg[e^{-\frac{a n}{N_0 T}I_{\Phi_{\text{LOS}}}^\text{GG}}\bigg] = e^{-2\pi \lambda p_\text{GG} \int_{0}^{\infty}\left(1 - \bbE_h\big[e^{-\frac{a n P_\text{t} A GG h}{x^\alpha N_0 T}}\big]\right)\ell_p(x) x \text{d}x }.\label{eq:singleLap}
 \end{equation}
  We invoke the MGF of a gamma random variable to yield the final Laplace transform of the PPP as
 \begin{equation}
 \cL_{I_{\Phi_{\text{LOS}}}^{\text{GG}}} =  e^{-2\pi \lambda  p_\text{GG}\int_{0}^{\infty}\big(1 -1 / (1+\frac{a n P_\text{t} A GG}{x^\alpha N_0 T N_\text{h}})^{N_\text{h}}\big) \ell_p(x) x \text{d}x }.\label{eq:laptrans}
 \end{equation}
 Each other Laplace transform is computed similarly but $ p_\text{GG}$ will correspond to the probability of the antenna gain $\{GG,Gg,gg\}$ and the NLOS probability is  $1-p_{\text{LOS}}$. We summarize our results in the following theorem.
 \begin{theorem}
 	The INR distribution of a mmWave \emph{ad hoc} network can be tightly bounded by
 	\begin{equation}
 	P^\mathrm{NL}(T) <  \sum_{n = 1}^{N_\text{C}} {N_\text{C} \choose n} (-1)^{n+1} e^{-2\pi \lambda (\zeta_\text{L} + \zeta_\text{N})} \label{eq:inrcdf}
 	\end{equation}
 	where
 	\begin{align}
 	\zeta_\text{L} &= \sum_{i} p_i\int_{0}^{\infty}\bigg(1 -1 / \big(1+\frac{a n P_\text{t} A M_i}{x^\alpha_\text{L} N_0 T N_\text{h}}\big)^{N_\text{h}}\bigg)  \ell_p(x) x \text{d}x \label{eq:R}\\
 	\zeta_\text{N} &= \sum_{i} p_i\int_{0}^{\infty}\bigg(1 -1 / \big(1+\frac{a n P_\text{t} A M_i}{x^\alpha_\text{N} N_0 T N_\text{h}}\big)^{N_\text{h}}\bigg) \big(1- \ell_p(x)\big) x \text{d}x \label{eq:U}
 	\end{align}
 	with $i\in\{\text{GG,Gg,gg}\}$.
 \end{theorem}
 
 \begin{proof}
 	Substituting the Laplace transform (\ref{eq:laptrans}) into (\ref{eq:binom}) yields the result. 
 \end{proof}
 
While we focus on investigating the impact of the node density and beamwidth of directional beamforming in this paper, the INR distribution also depends on other system parameters, such as transmission power. It should be noted that the INR in \eqref{eq:inr} scales with the transmit power; interesting future work is discovering a transmission power control scheme to optimize the INR. Such a scheme could limit the transmit power based on the proximity of the nearest interferer.
 
 \subsection{One-Way Performance Analysis}
  Now, using Theorem 1, we characterize the transmission capacity, $\lambda_\epsilon$. This is the largest $\lambda$ the network can support given an SINR threshold, $T$ and outage $\epsilon$. More simply, $1-\epsilon = P_\text{c}(T)$ of users will have an SINR larger than $T$. The transmission capacity can also be defined as the number of successful transmissions per unit area, which is directly connected to the number of users supported by the network. To do this, we approximate the exponential terms of Theorem 1 as
   \begin{equation}\label{eq:onewaytx}
   P_\text{c}^{\text{L}} < \sum_{n=1}^{N_\text{h}} (-1)^{n+1}{N_\text{h} \choose n}  e^{-nK_\text{L}TN_0}\bigg(1-2\pi\lambda_\epsilon \Theta+2\pi\lambda_\epsilon^2\Theta^2\bigg)
   \end{equation} %e^{-nKTN_0}
  where $\Theta = \kappa_\text{L} + \kappa_\text{N} $. We leverage the bound, $e^{-x} \leq (1-x+x^2/2)$ for $x \in\bbR^+$, for the Laplace functional term. This bound is tight for small $x$. We are interested in analyzing the optimal $\lambda$ for $P_\text{c}$ near $1$. As a result, the Laplace functional will be close to 1; the argument will be close to 0. A similar bound is done for the NLOS term in Theorem 1. We combine (\ref{eq:onewaytx}) and the NLOS approximation to form
  \begin{align}
  \begin{split}\label{eq:onewaytxFULL}
  1-\epsilon&< \sum_{n=1}^{N_\text{h}} (-1)^{n+1}{N_\text{h} \choose n}  e^{-nK_\text{L}TN_0}\times\\
  &\bigg(1-2\pi\lambda_\epsilon \Theta+2\pi\lambda_\epsilon^2\Theta^2\bigg) + \sum_{n=1}^{N_\text{h}} (-1)^{n+1}{N_\text{h} \choose n}\times\\
  &e^{-nK_\text{N}TN_0}\bigg(1-2\pi\lambda_\epsilon \Psi+2\pi\lambda_\epsilon^2\Psi^2\bigg)
  \end{split} 
  \end{align}%e^{-nKTN_0} 
with $\Psi = \xi_\text{L} +\xi_\text{N}  $. Because of this bound, $P_\text{c}$ is now a quadratic equation in $\lambda$ which can be solved in closed-form. The exact solution depends on $N_\text{h}$. Symbolic tools, such as \emph{Mathematica}, can factor and solve (\ref{eq:onewaytxFULL}) such that
\begin{equation}
\lambda_\epsilon = f(T,\epsilon).\label{eq:txcap}
\end{equation}
  
Area spectral efficiency is a useful metric because it can characterize the network performance, rather than just a single link, as SINR does \cite{AndGan10}. We define area spectral efficiency as
  \begin{equation}
  \text{ASE} := \underbrace{\lambda_\epsilon}_{\frac{\text{users}}{\text{area}}}\underbrace{\text{log}_2(1+T)}_\text{efficiency} \underbrace{(1-\epsilon)}_\text{\% of the time}.\label{ASE}
  \end{equation}
  Substituting \eqref{eq:txcap} into (\ref{ASE}) yields a function of just $T$ and $\epsilon$. The ASE yields a result in terms of bits/sec/Hz/m$^2$.

\section{Two-way Ad Hoc Communication}\label{sec:twoway}

The derivations from the Section \ref{sec:oneway} are for \emph{one-way} communication. There is no consideration for the reverse link (i.e. receiver to transmitter). In real systems, however, successful transmission usually relies on a two-way communication stack. The two-way transmission capacity quantifies the maximum density of users a network can support while \emph{both} the forward and reverse link are subject to outage constraint, $\epsilon$ \cite{VazTru11}.

The \emph{forward} link is defined as the transmitter to receiver link (i.e. what was discussed in Section \ref{sec:oneway}), while the \emph{reverse} link is the receiver to transmitter control link. Frequency division duplexing (FDD) is used between the forward and reverse links, as is done in \cite{VazTru11}. Each link operates concurrently with differing rate requirements. Consider the bandwidth from Section \ref{sec:oneway} split among the forward and reverse links. Hence, $B_{\text{total}}$ is the bandwidth available to the system. The forward link is allocated $B_{\text{F}}$, while the reverse link is allocated $B_{\text{R}} = B_{\text{total}}- B_{\text{F}}$. The SINR is similarly defined as $\text{SINR}_\text{F}$ and $\text{SINR}_\text{R}$. Correspondingly, from Shannon's equation, the links achieve rates, $R_\text{F}$ and $R_\text{R}$. A user with rate requirement $R_\text{F}$ would then have an SINR threshold of $T_\text{F} = 2^{R_\text{F}/B_\text{F}}-1$. It should be noted that time division duplexing can similarly be used with the threshold of $T_\text{F} = 2^{\frac{R_\text{F}}{\tau_\text{F}B_\text{total}}}-1$ with $\tau_\text{F}$ being the fraction of time for the forward link. The reverse link thresholds are similarly defined. We consider only FDD for the remainder of the analysis.

\subsection{Two-way SINR Analysis}
The two-way SINR probability is the probability that the forward link \emph{and} reverse link exceed an SINR threshold. More precisely,

\begin{equation}
P_\text{c}^\text{tw} = \bbP[\mathrm{SINR_F} > T_\text{F}, \mathrm{SINR_R} > T_\text{R}].
\end{equation}

We assume that the forward and reverse link do not have the same SINR threshold because the reverse control link is generally low-rate compared to the forward link. To analyze this probability, we leverage the following definitions and lemma.

\emph{Definition 1}\cite{VazTru11}: A random variable $X$ defined on $(\Omega,\cF,\bbP)$ is increasing if $X(\omega) \leq X(\omega')$ for a partial ordering on $\omega, \omega'$. $X$ is decreasing if $-X$ is increasing.

The SINR is a random variable defined on the probability space which is determined by how the interferers are placed on the plane. Let $\omega$ be a set of active interferers from the PPP. Then, $\omega' \geq \omega$ if $\omega'$ is a superset of $\omega$. The SINR \eqref{eq:SINR} decreases if another interferer is added: $\text{SINR}(\omega) \geq \text{SINR}(\omega')$. Therefore, SINR is a decreasing random variable. 

\emph{Definition 2}\cite{VazTru11}: An event $A$ from $\cF$ is increasing if $\bbI_A(\omega) \leq \bbI_A(\omega')$ when $\omega \leq \omega'$ where $\bbI_A$ is the indicator function. The event is decreasing if $A^c$ is increasing.

The SINR probability event, $\{\text{SINR}>T\}$ is a decreasing event. If another interfering user is added to $\omega$, the probability of successful transmission decreases. Now, we can leverage the Fortuin, Kastelyn, Ginibre (FKG) inequality \cite{fkg}.

\emph{Lemma 2}\cite{fkg}: If both $A, B \in \cF$ are increasing or decreasing events then $P(AB)\geq P(A)P(B)$.

The FKG inequality can give a bound on the two-way SINR probability. The bound is only tight when the forward and reverse channels are independent; the dependence, however, can be low in \emph{ad hoc} network as shown in \cite{VazTru11,Bai2014}. In \cite{VazTru11}, this was shown to be a tight lower bound. Using FKG, we can define the two-way SINR probability as

\begin{equation}
P_\text{c}^\text{tw} \geq \bbP[\mathrm{SINR_F} > T_\text{F}]\bbP[\mathrm{SINR_R} > T_\text{R}].
\end{equation}
Therefore, the two-way SINR probability can be lower-bounded as
\begin{align}\label{twowayPc}
\begin{split}
P_\text{c}^\text{tw} \geq& \Bigg[\sum_{n=1}^{N_\text{h}} (-1)^{n+1}{N_\text{h} \choose n} e^{-2\pi \lambda \big[\kappa_\text{L}(T_\text{F})+\kappa_\text{N}(T_\text{F})\big]}\Bigg] \\
&\times\Bigg[\sum_{n=1}^{N_\text{h}} (-1)^{n+1}{N_\text{h} \choose n} e^{-2\pi \lambda  \big[\kappa_\text{L}(T_\text{R})+\kappa_\text{N}(T_\text{R})\big]}\Bigg].
\end{split}
\end{align}

\subsection{Two-Way Performance Analysis}
Now we compute the two-way transmission capacity, $\lambda_\epsilon^{\text{tw}}$. Because of the constraint that both transmitter and receiver must succeed in transmission, we can argue $\lambda_\epsilon^{\text{tw}} \leq \lambda_\epsilon$. It is unclear, however, if the gap is large in a mmWave network. Using the transmission capacity framework can quantify how many users must be removed from the network to support the reverse link. Using a similar approach as with the \emph{one-way} transmission capacity, we use a Taylor expansion of the exponential function to yield
\begingroup\makeatletter\def\f@size{8}\check@mathfonts
\def\maketag@@@#1{\hbox{\m@th\normalsize\normalfont#1}}%
\begin{align}\label{twowaytx}
\begin{split}
P_\text{c}^\text{tw} \approx& \Bigg[\sum_{n=1}^{N_\text{h}} (-1)^{n+1}{N_\text{h} \choose n} \bigg(1-2\pi\lambda_\epsilon^{\text{tw}} \Theta(T_\text{F})+2\pi(\lambda_\epsilon^{\text{tw}})^2\Theta^2(T_\text{F})\bigg)\Bigg] \times \\
&\Bigg[\sum_{n=1}^{N_\text{h}} (-1)^{n+1}{N_\text{h} \choose n} \bigg(1-2\pi\lambda_\epsilon^{\text{tw}} \Theta(T_\text{R})+2\pi(\lambda_\epsilon^{\text{tw}})^2\Theta^2(T_\text{R})\bigg)\Bigg].
\end{split}
\end{align}\endgroup
The result is a quartic equation in $\lambda_\epsilon^{\text{tw}}$ which has an analytic expression. The general solution, however, is quite messy, and the equation is a page long, so it is omitted here. An analytical solver, such as \emph{Mathematica}, can factor the coefficients of \eqref{twowaytx} which can be input into a polynomial root solver to yield the solution. The two-way area spectral efficiency can be defined as \cite{VazHea12}
\begin{equation}\label{eq:asetw}
\text{ASE}_\epsilon^\text{tw} := \lambda_\epsilon \bigg( \frac{R_\text{F} + R_\text{R}}{B_\text{total}}\bigg)(1-\epsilon).
\end{equation}
Given rate requirements $R_\text{F}$ and $R_\text{R}$, what is the allocation of bandwidth that maximizes \eqref{eq:asetw}? We explore this trade-off in Section \ref{sec:results}. 

\section{Performance Results}\label{sec:results}
In this section, we evaluate the performance metrics to obtain the transmission capacity, $\lambda_\epsilon$. Further, we compute the area spectral efficiency to define the \emph{best} $\lambda$, given by $\lambda^\star$. We compare the achievable rates for mmWave networks with classic results for lower frequency \emph{ad hoc} networks. The section is concluded with an investigation into two-way communication. 

Throughout the section, we compare the mmWave results to UHF \emph{ad hoc} networks (e.g. 2.4 GHz). For the UHF network, we adjust the model parameters to fit UHF networks. We maintain a constant antenna \emph{aperature} between models which keeps the relative physical size of the devices constant. For an antenna, the gain is computed using $G = \frac{A_\text{eff}}{\lambda^2 / 4\pi}$ where $A_\text{eff}$ is the aperature of the antenna. By increasing the frequency ten-fold (e.g. 2.4GHz to 28GHz), the gain of the resulting mmWave antenna is 100 (20dB); this matches our 20dB total gain for both transmitter and receiver (i.e. 10dB for each transmitter and receiver). We maintain 1W (0dB) of transmit power for UHF. To capture the effect of LOS/NLOS communication, we use $\alpha_\text{L} = 2.09$ and $\alpha_\text{N} = 3.75$ as shown in \cite{GaliottoGMD14} which are taken from 3GPP LTE measurements. We use the same blockage model as mmWave. We use a path-loss intercept of $40.4$dB and a noise power of $-127$dB (e.g noise power for 50MHz). For the rate calculations, we use a bandwidth of 50MHz.

\subsection{Transmission Capacity}  
  \begin{figure}
  \centering
  \begin{subfigure}{.5\textwidth}
  \centering
  \includegraphics[width=.65\linewidth]{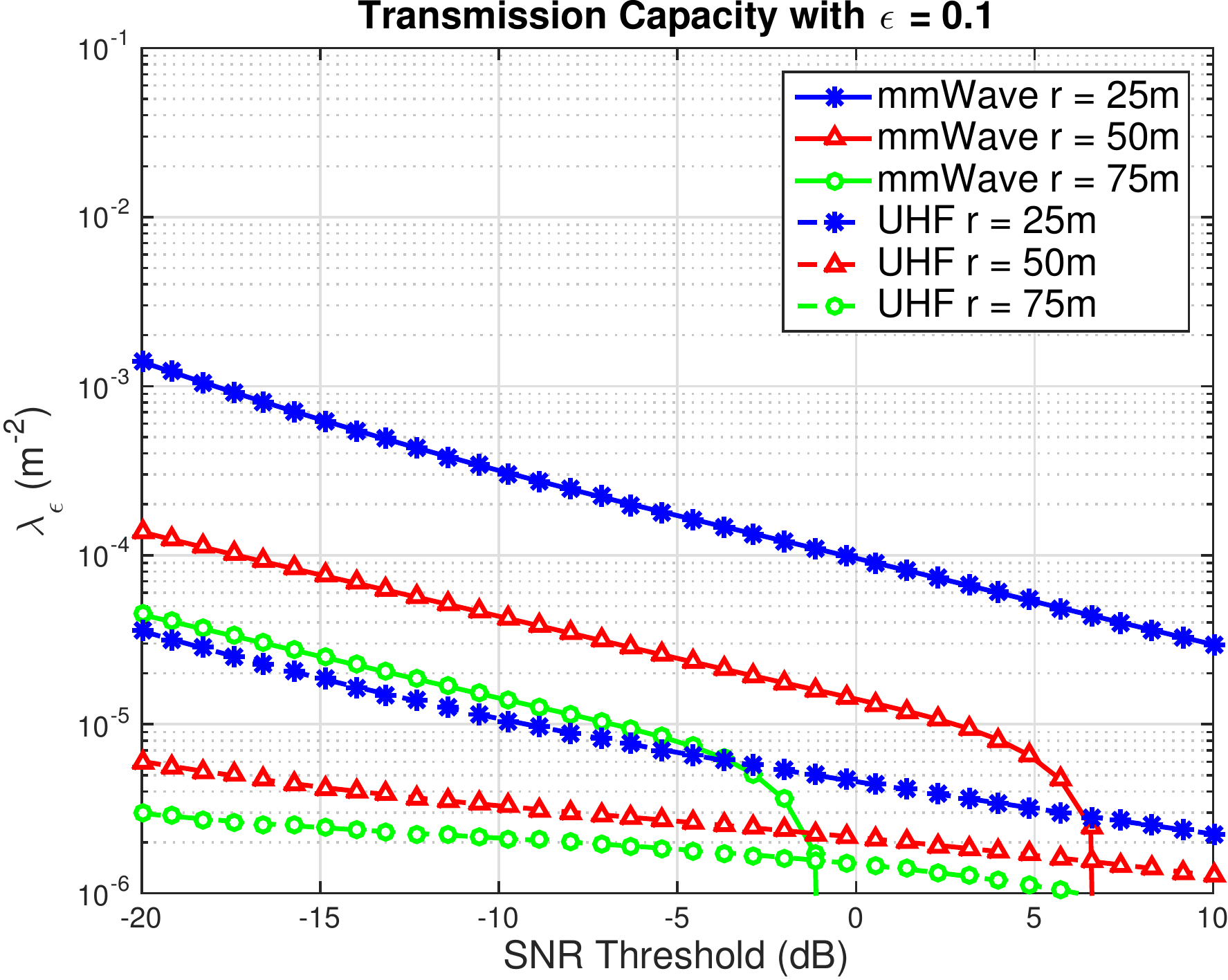}
 \caption{}
  \label{fig:txcap_sub1}
  \end{subfigure}
  \begin{subfigure}{.5\textwidth}
  \centering
  \includegraphics[width=.65\linewidth]{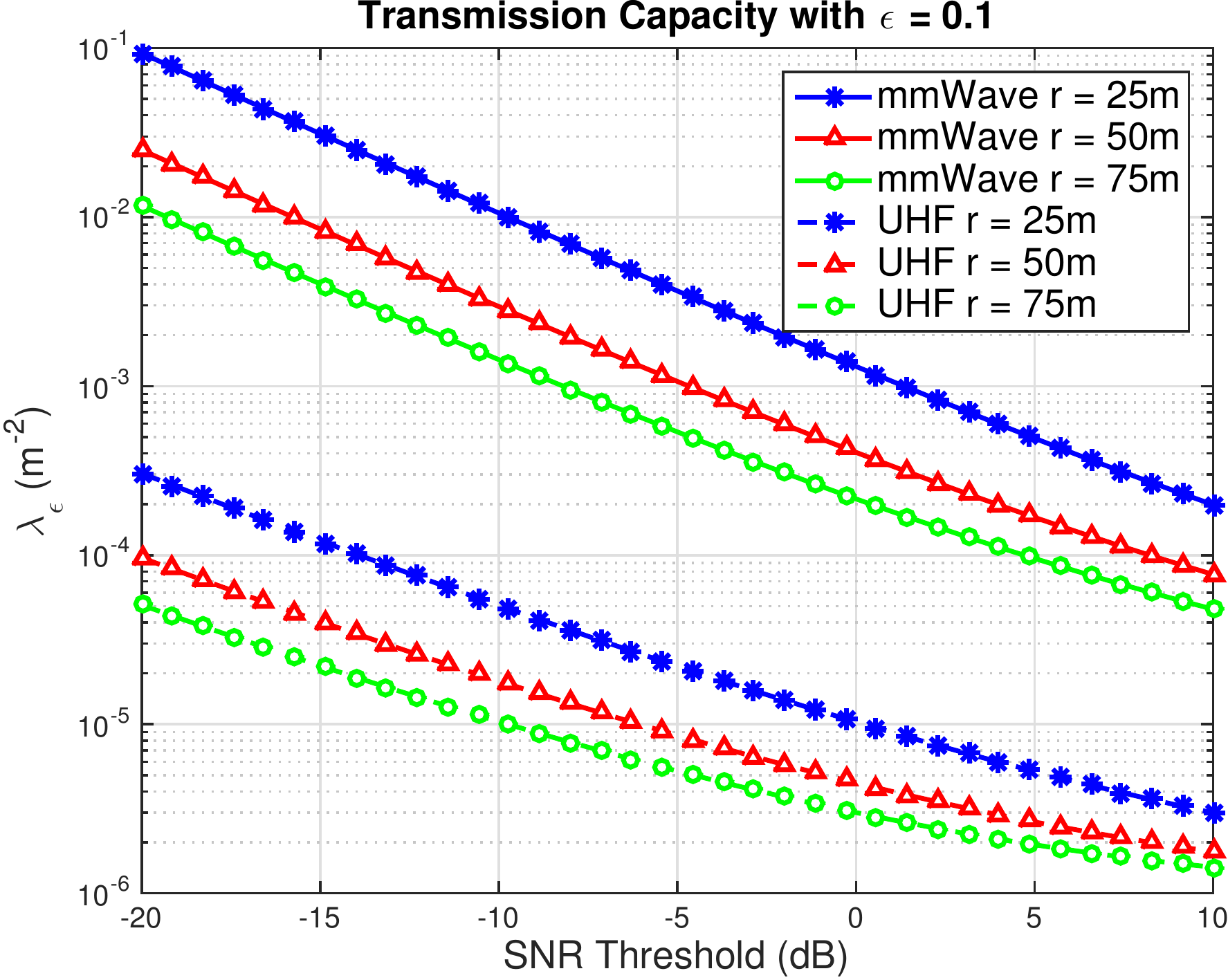}
  \caption{}
  \label{fig:txcap_sub2}
  \end{subfigure}
  	\caption{The largest $\lambda$ for a 10\% outage at various SINR thresholds and dipole distances for NLOS/LOS communication (a) and LOS-only communication (b). Note the different y-axis scales.}
   	\label{fig:epslam1}
  \end{figure}
  
Fig. \ref{fig:epslam1} shows the transmission capacity for mmWave and lower frequency networks with a 10\% outage. Fig. \ref{fig:epslam1} shows the relationship between providing a higher SINR (and thus rate) to users while maintaining a constant outage constraint. As expected, the shortest dipole length can support the highest density of users. A linear increase in SINR (in dB) results in an exponential decrease in the density of users in the network. 

In Fig. \ref{fig:txcap_sub1}, both LOS and NLOS communication is allowed. If the dipole length is 25m, mmWave networks can allow a larger density. If the dipole length is 50m or 75m, however, lower-frequency networks can permit higher densities when the communication threshold is higher. This is because the mmWave network begins to be noise limited. Essentially, the blockage probability is larger than $\epsilon$; because of the longer link length (and increased path-loss exponent for NLOS communcation), there is no density that will meet the threshold requirements and the transmission capacity is 0. For the UHF network, the lower path-loss exponent and noise power permit a positive transmission capacity. Fig. \ref{fig:txcap_sub2} shows the improvement if communication is kept to LOS links. Because the communication is always LOS, the longer links can now support a positive transmission capacity for higher SINR thresholds.

\subsection{Area Spectral Efficiency}
 %With this, we plot the ASE for two outage constraints
 
   \begin{figure}
   \centering
   \begin{subfigure}{.5\textwidth}
   \centering
   \includegraphics[width=.65\linewidth]{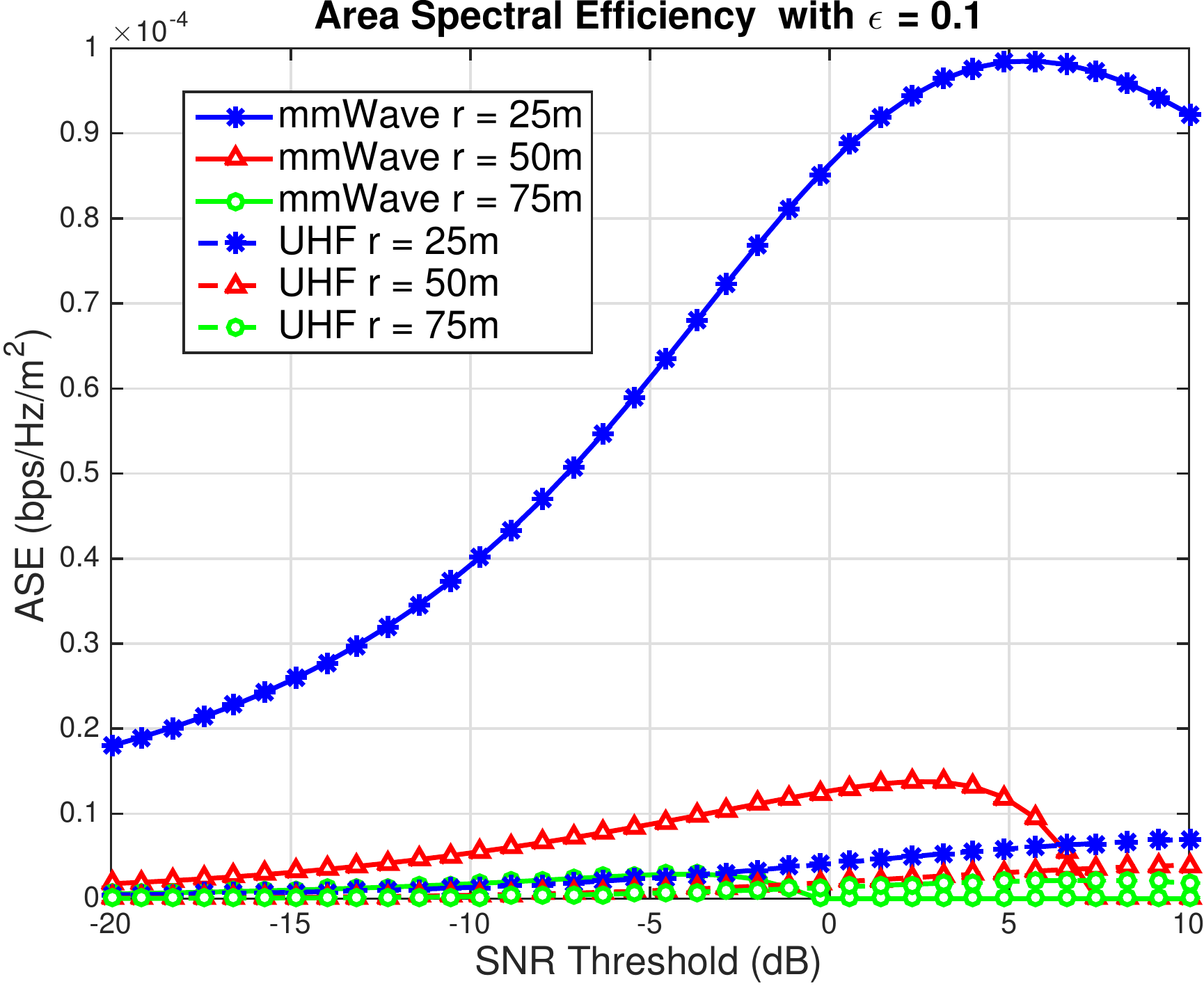}
  \caption{}
   \label{fig:ase_sub1}
   \end{subfigure}
   \begin{subfigure}{.5\textwidth}
   \centering
   \includegraphics[width=.65\linewidth]{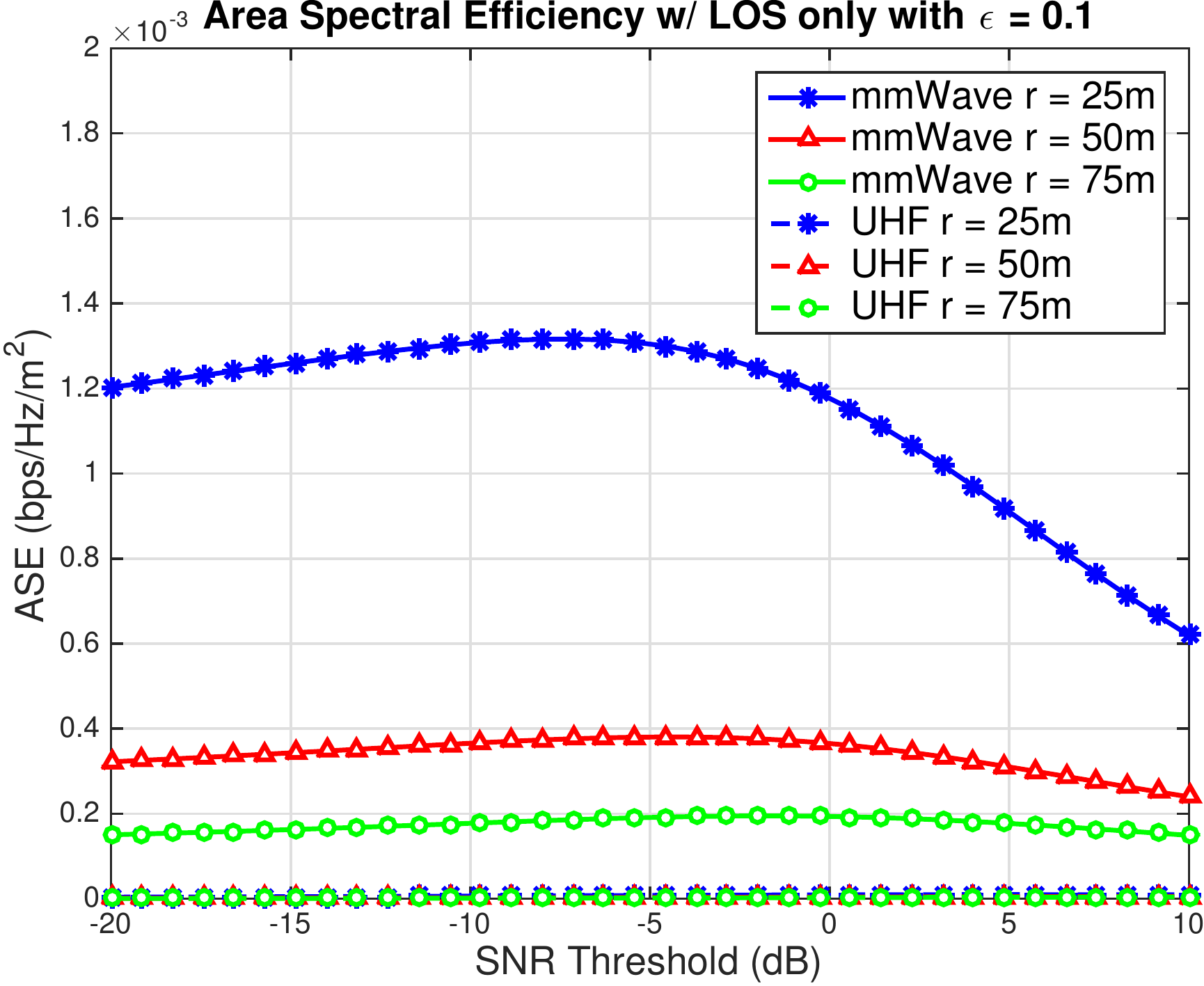}
   \caption{}
   \label{fig:ase_sub2}
   \end{subfigure}
	\caption{Area spectral efficiency of network with 10\% outage. If the dipole link is restricted to LOS (b), an order-of-magnitude improvement is shown over NLOS/LOS dipole links (a).}
 	\label{fig:eps1}
   \end{figure}

Similar trends are evident in Fig. \ref{fig:eps1}. The mmWave network has a 10$\times$ efficiency gain compared to UHF networks when the transmission capacity is non-zero. This gain is realized through the interference reduction in the directional antenna array and the increased path-loss exponent for NLOS links. Because buildings do not attenuate UHF as much, even the NLOS interference in a UHF network limits performance.

The shape of the curves suggests an optimal density with respect to ASE. This leads to the optimization problem
\begin{equation}
\lambda^\star = \underset{\lambda_\epsilon}{\operatorname{argmax}}  \hspace{2mm}\lambda_\epsilon\text{log}_2(1+T)(1-\epsilon). \label{lamstar}
\end{equation}
The numerical solution to this problem is the density corresponding to the largest ASE from Fig. \ref{fig:eps1}. We leave the exploration of analytical solutions to \eqref{lamstar} for future work. Fig. \ref{fig:opt} shows the numerically obtained $\lambda^\star$ from Fig. \ref{fig:ase_sub2}. The optimal density is exponentially decreasing in $r$. The optimal density, $\lambda^\star$, corresponds to an average neighbor distance 1/2 the link distance in the LOS-only (\emph{protocol gain}) case. \textbf{MmWave \emph{ad hoc} networks can not only support high density, but this density is best for overall network efficiency}. This is due to both the directional antennas and blockage. The blockage thins the interference PPP as shown in Section \ref{sec:LOSint}. The remaining LOS interferers are effectively \emph{pushed away}. The interference power from a close neighbor into the side-lobe (i.e. the power is heavily attenuated) is the same as that interferer being further away but using omni-directional antennas. Of course, if an interferer is in the main-lobe of the antenna, this phenomenon works against the receiver, but more often, it helps.

\begin{figure}
\centering
\includegraphics[width=0.65\columnwidth]{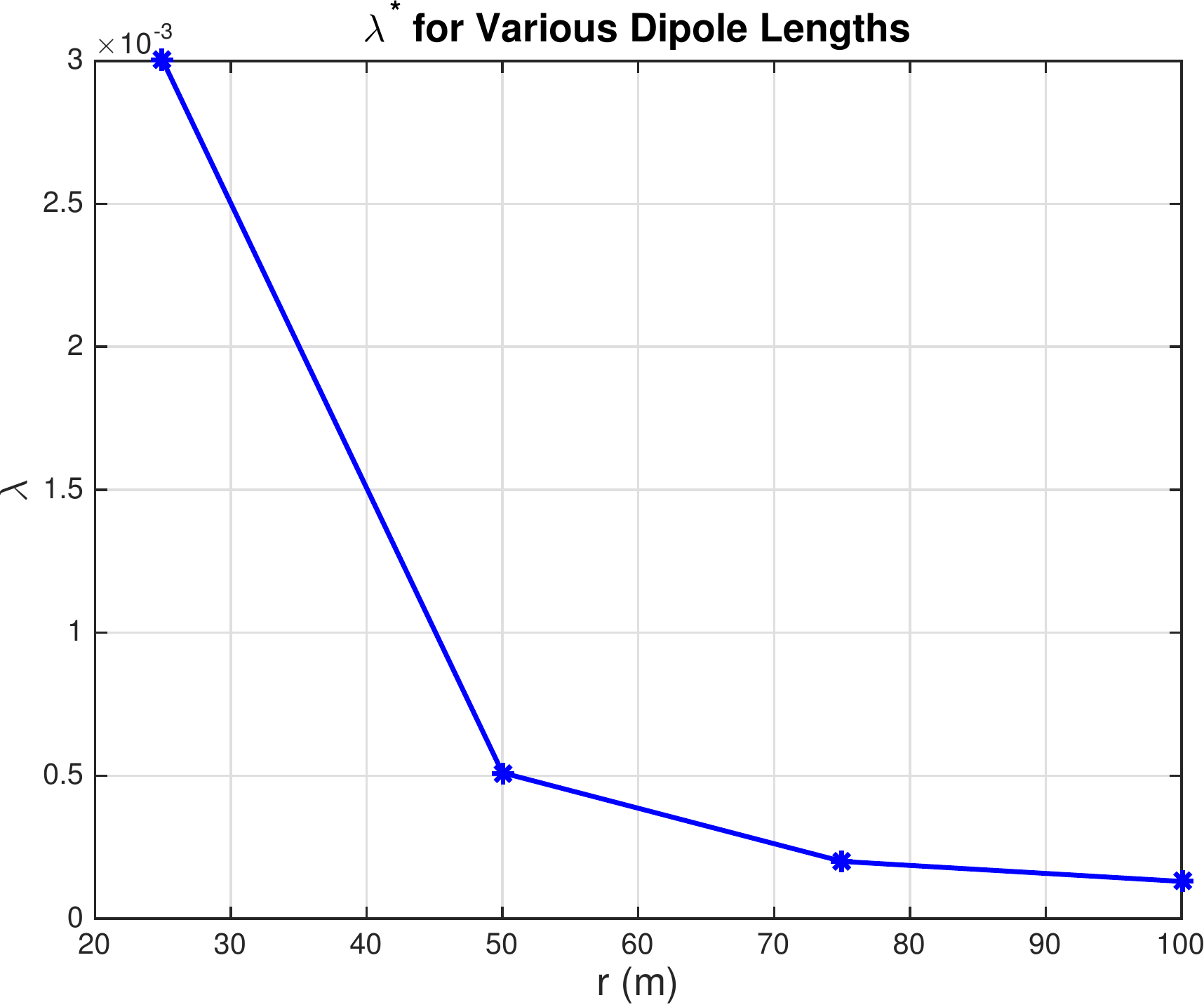}
	\caption{Optimal network density for various dipole lengths, subject to 10\% outage.}
 	\label{fig:opt}
\end{figure}
  %In \cite{Bai2013b}, it was shown that $N = 3$ yields a close approximation to the non-fading case. Additionally, the noise is negligible in all but the very sparse networks. Thus, the coverage probability is
%\begin{multline}
%    P_\text{c} = \sum_{n=1}^{3} (-1)^{n+1} \times \\
%    \bigg(1-2\pi\lambda \Phi(n,T)+4\pi\lambda^2\Phi^2(n,T)\bigg)    
%\end{multline}

 \subsection{Rate Analysis}
Fig. \ref{fig:rate} shows the rate coverage probability, where $R = W \text{log}_2(1+T)$, and $W$ is the system bandwidth. From Theorem 1, a user will achieve SINR $>T$ with some probability as shown in Fig. \ref{fig:sub_lam_1} and Fig. \ref{fig:sub_lam_2} which leads to an achievable rate probability. For example, according to Fig. \ref{fig:sub_lam1_los}, a LOS mmWave communication link of 50m will have an SINR of at least 10dB 95\% of the time which, assuming Gaussian signaling, leads to a rate according to Shannon's equation. In Fig. \ref{fig:rate} we consider networks with both LOS and NLOS communication.
 
 \begin{figure}
 \centering
 \includegraphics[width=0.65\columnwidth]{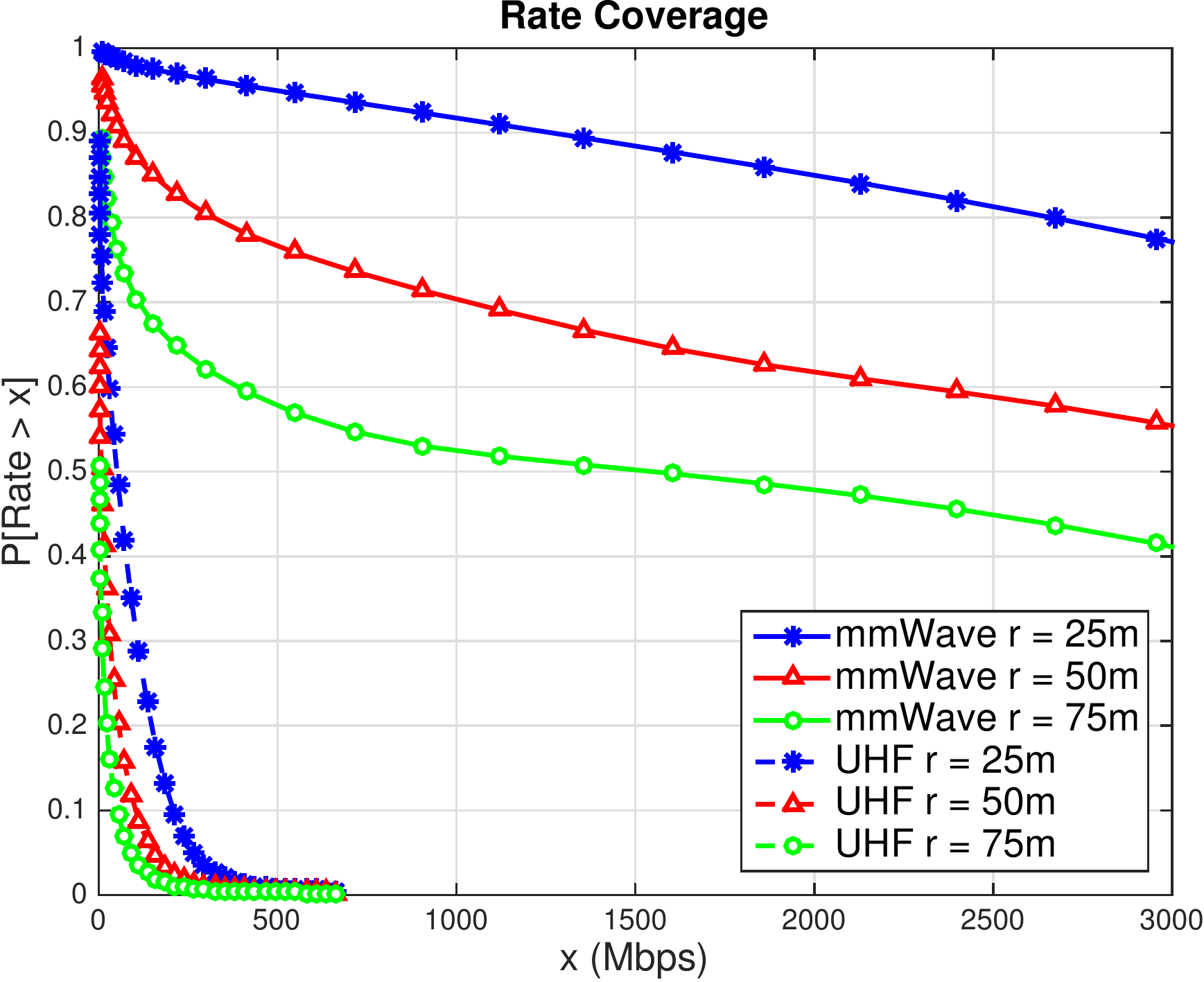}
 	\caption{mmWave \emph{ad hoc} networks provide significant increase in rate coverage over lower frequency networks.}
  	\label{fig:rate}
 \end{figure}
 
 The system bandwidth used in Fig. \ref{fig:rate} is 500MHz for the mmWave and 50MHz for the lower frequency system. While the bandwidth is only a 10$\times$ increase, we see orders-of-magnitude increase in the rate coverage for mmWave networks. All link lengths of the mmWave network support over 1Gbps a majority of the time.
 
 \subsection{INR Distribution}
 Figs. \ref{fig:noise_9}, \ref{fig:noise_30}, and \ref{fig:noise_90} show the INR CDF for three values of $\lambda$ for each of the beam patterns in Fig. \ref{fig:antenna_pic}. Indeed, in all antenna patterns, the sparsest network exhibits noise limited behavior. For example, the $\bbP[\text{INR} < 0\text{dB}] = 0.4$ for $30^\circ$ antennas in the sparest network. Yet, these results show compelling evidence that a mmWave \emph{ad hoc} network can still be considered interference limited. In dense networks (22m and 70m spacing), in all but the very narrow beam case, the network exhibits strong interference. Because of this, we urge caution when considering mmWave networks to be noise limited.
 %   \begin{figure}
 %   	\centering
 \begin{figure}
 	\centering
 	\includegraphics[width=0.65\columnwidth]{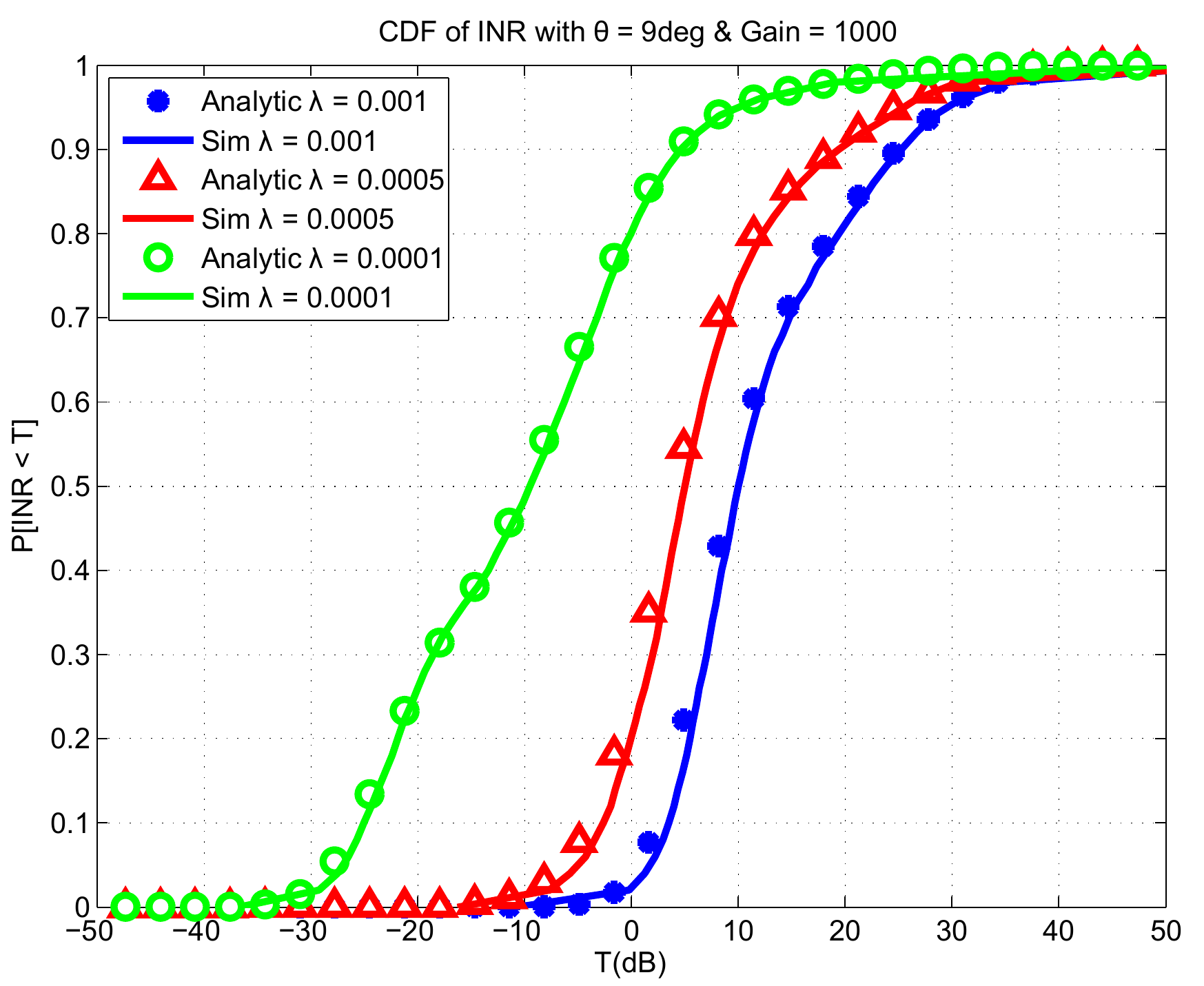}
 	\caption{The INR CDF for $\theta = 9^\circ$. With extreme beamforming, the network remains interference limited in all but the sparest network.}\label{fig:noise_9}
 \end{figure}
 \begin{figure}
 	\centering
 	\includegraphics[width=0.65\columnwidth]{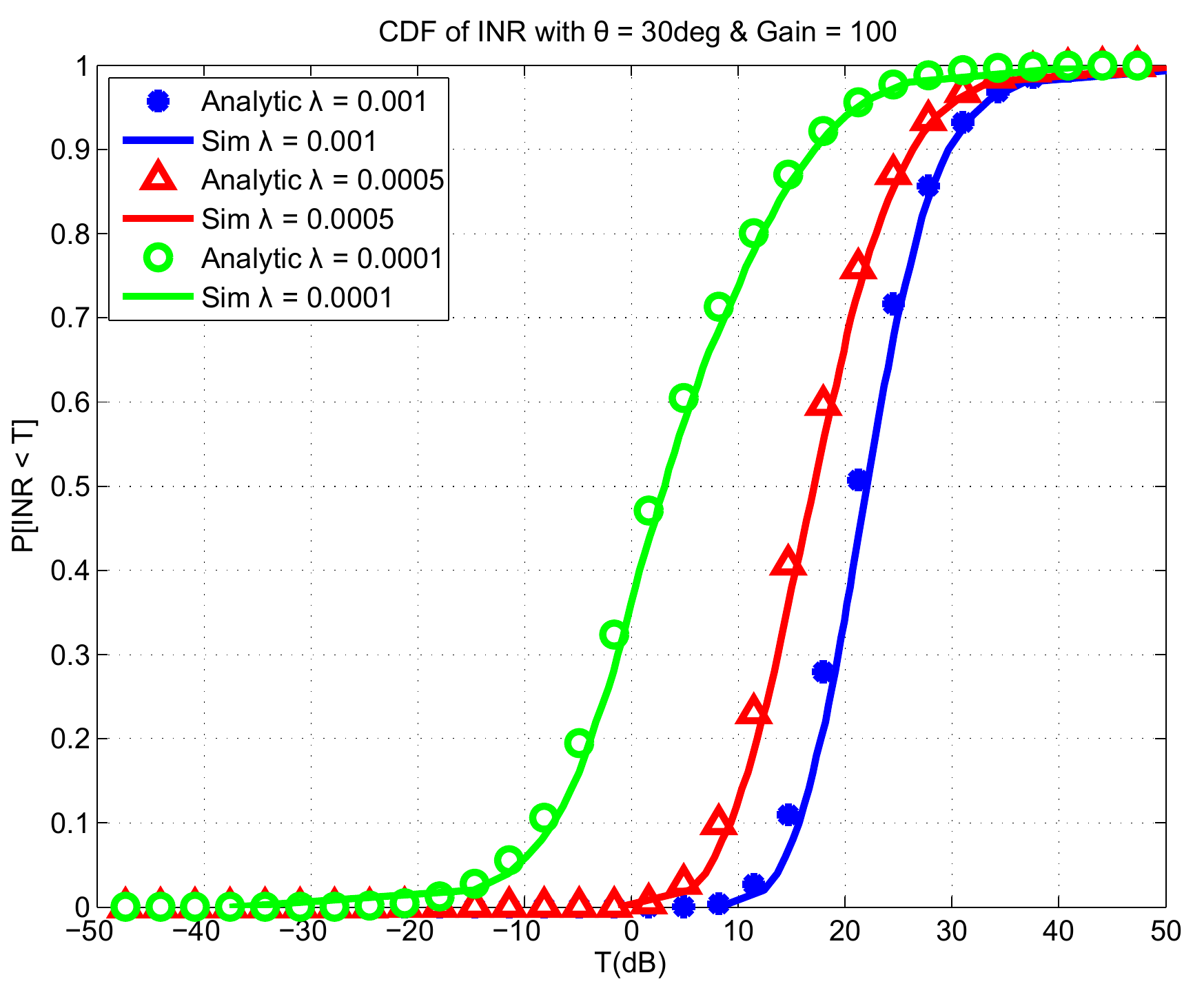}
 	\caption{The INR CDF for $\theta = 30^\circ$. In the sparsest network, the interference power is more dominant than the noise power (i.e. $\bbP[\mathrm{INR} < 0 \mathrm{dB}] = 0.4$ for the green circle network), but the red triangle curve shows that the more dense network is always interference limited.}\label{fig:noise_30}
 \end{figure}
 \begin{figure}
 	\centering
 	\includegraphics[width=0.65\columnwidth]{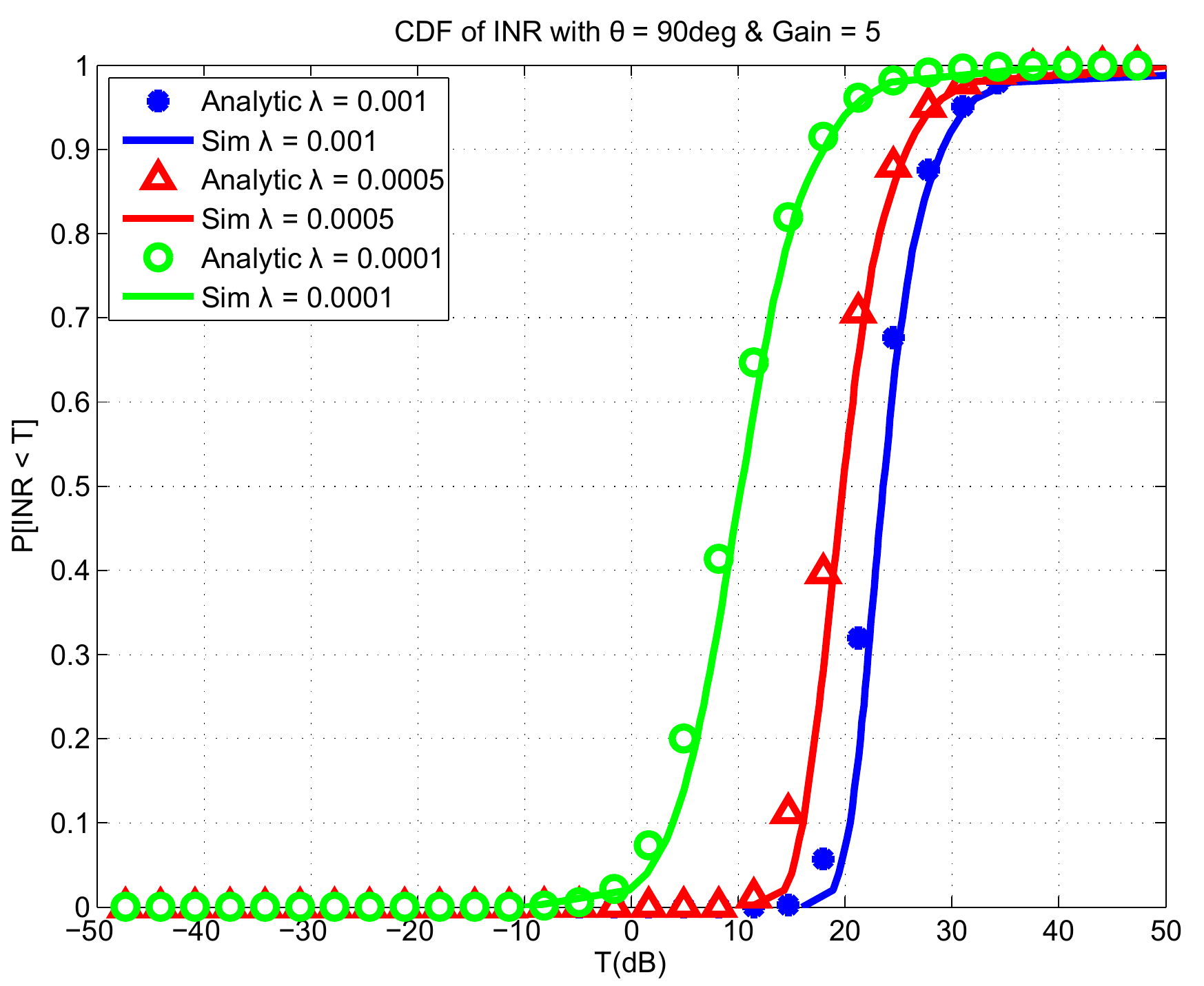}
 	\caption{The INR CDF for $\theta = 90^\circ$. In all networks, the interference power is nearly always more dominant than the noise power (i.e. $\bbP[\mathrm{INR} < 0 \mathrm{dB}] = 0.05$ for the green circle network).}	\label{fig:noise_90}
 \end{figure}
 %   	\caption{The instantaneous INR CDF for a mmWave \emph{ad hoc} network. In all but the sparsest networks with extremely high gain antennas, the network is interference limited.}
 %   	\label{fig:noise_cdf}
 %   \end{figure}     
 
 Fig \ref{fig:losonly} shows the INR distribution if we ignore NLOS interference for when $\theta = 30^\circ$. It shows that for many mmWave networks the interference is largely driven by the LOS interference in the two denser networks. The CDF of the two denser networks in Fig. \ref{fig:losonly} is nearly identical to Fig. \ref{fig:noise_30} which indicates that NLOS interference plays no role at those densities. We believe this shows compelling evidence that interference cancellation may be useful, even at mmWave frequencies. In particular, eliminating LOS interference is most important.
 \begin{figure}
 	\centering
 	\includegraphics[width=0.65\columnwidth]{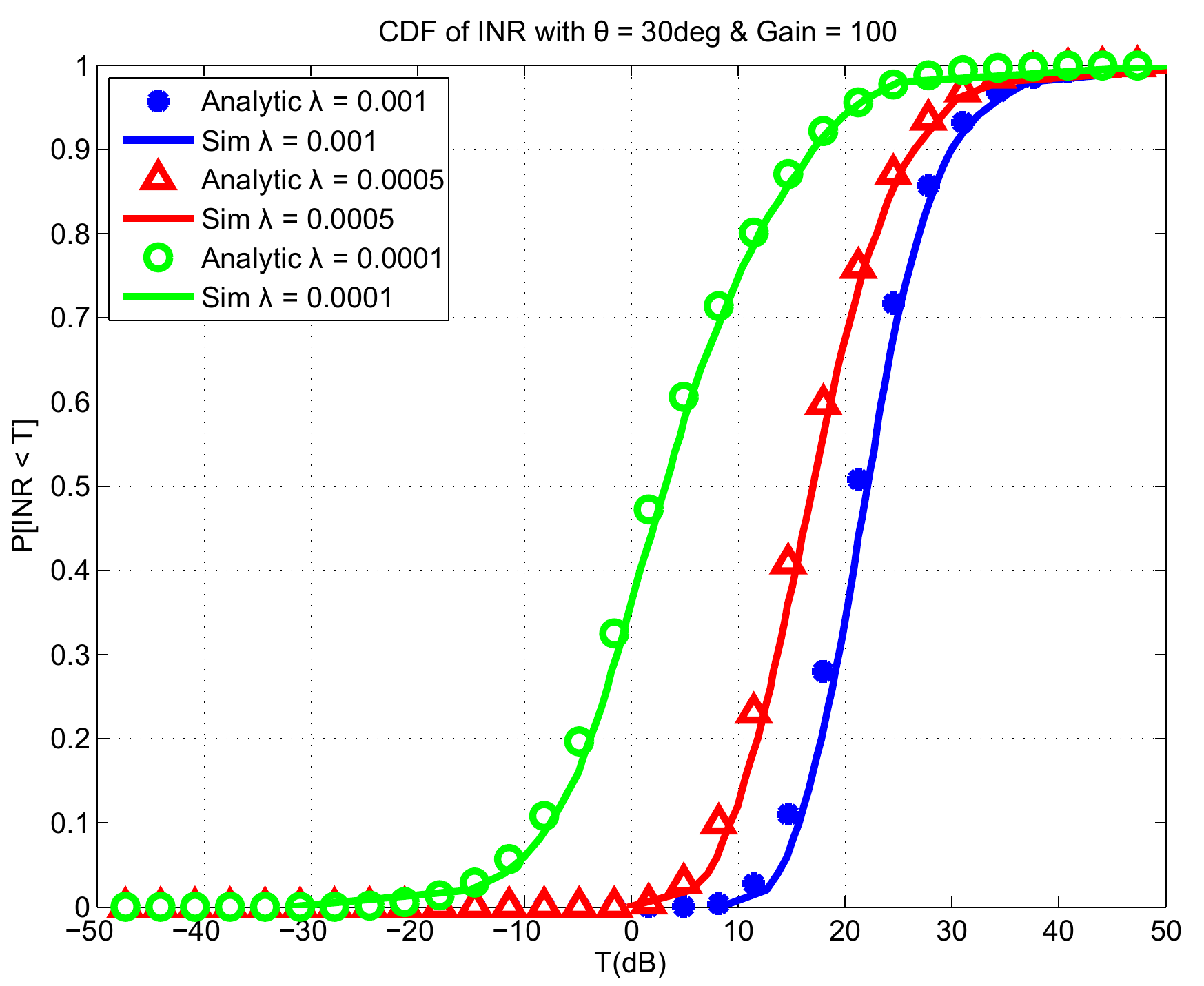}
 	\caption{The INR CDF for $\lambda = 5\times 10^{-5}$ and $\theta = 30^\circ$ with only LOS interference. Compare to Fig. \ref{fig:noise_30}, we find that the shape of INR distributions is largely determined by the LOS interference when the network is dense.}
 	\label{fig:losonly}
 \end{figure}
 
 In Fig. \ref{fig:inrtxcap}, the INR is shown for the transmission capacity of the networks from Figs. \ref{fig:txcap_sub1} \& \ref{fig:txcap_sub2}. If conditioned on LOS communication (i.e. \emph{LOS protocol-gain}), the networks support very dense deployments. As such, the INR is nearly always $>0$dB as shown in Fig. \ref{fig:inrtxcap}. If the network does not enforce a LOS-only transmission scheme, the transmission capacity is less. The interference, however, is not negligible for networks of 25m and 50m. If the communication link is 25m, the INR is $>0$dB 70\% of the time; if the link is 50m, the INR is less but still $>-10$dB roughly half the time.
  \begin{figure}
  	\centering
  	\includegraphics[width=0.65\columnwidth]{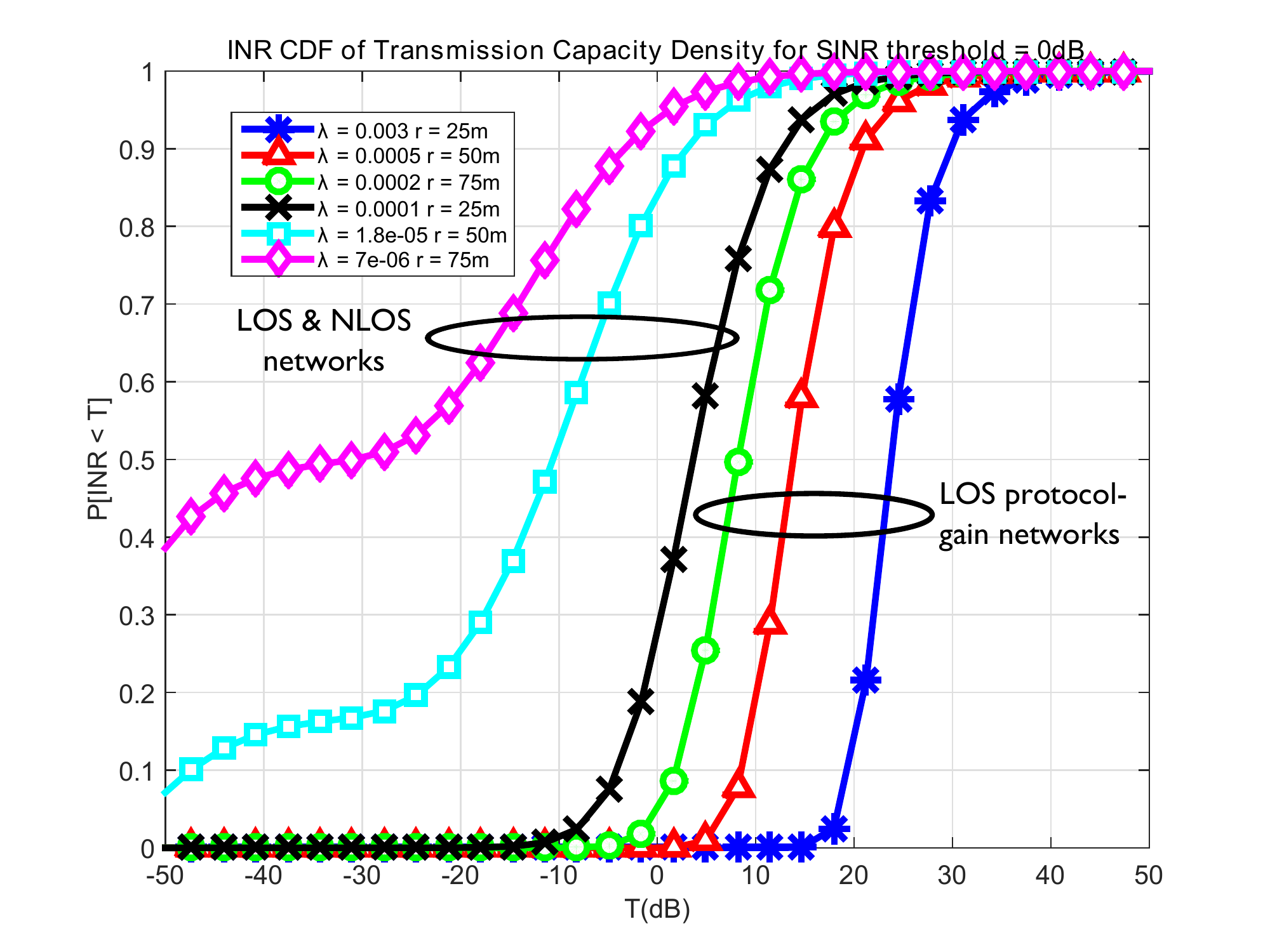}
  	\caption{The densities correspond to the transmission capacity from Figs. \ref{fig:txcap_sub1} \& \ref{fig:txcap_sub2} for SINR threshold on 0dB.}
  	\label{fig:inrtxcap}
  \end{figure}
 
  \subsection{Two-Way Communication Results}
  
  The results presented in this section consider a two-way system using bandwidth allocation to split resources. We show that, in asymmetric traffic, the transmission capacity of a two-way network can be vastly improved compared to equal bandwidth allocation or rate-proportional allocation. The two-way area spectral efficiency is compared to one-way area spectral efficiency. We show that 75\% of the one-way efficiency can be achieved for outage of 10\% which is a 100\% increase over the baseline equal allocation. In all the results, the dipole link length is 50m.  
  
 % \subsection{Comparison of One-way and Two-way}
  
 % Consider a future \emph{tactile} internet game application where the rate requirements between the forward and reverse link are symmetric.
  
%\begin{figure}
%   \centering
%   \includegraphics[width=0.65\columnwidth]{one_two_tx.pdf}
%   	\caption{In a two-way network, the transmission capacity is less than in a one-way network.}
%    	\label{fig:twoone}
%\end{figure}
 
 We consider asymmetric traffic. For example, in TCP assuming 1000 byte data packets, the receiver must reply with 40 byte $ACK$ packets \cite{tcp}. Hence, the rate asymmetry in TCP is $1/25$. The following results consider a system bandwidth of 100MHz, a forward rate requirement of 200Mbps, and a reverse link rate requirement of 8Mbps.
   
\begin{figure}
      \centering
      \includegraphics[width=0.65\columnwidth]{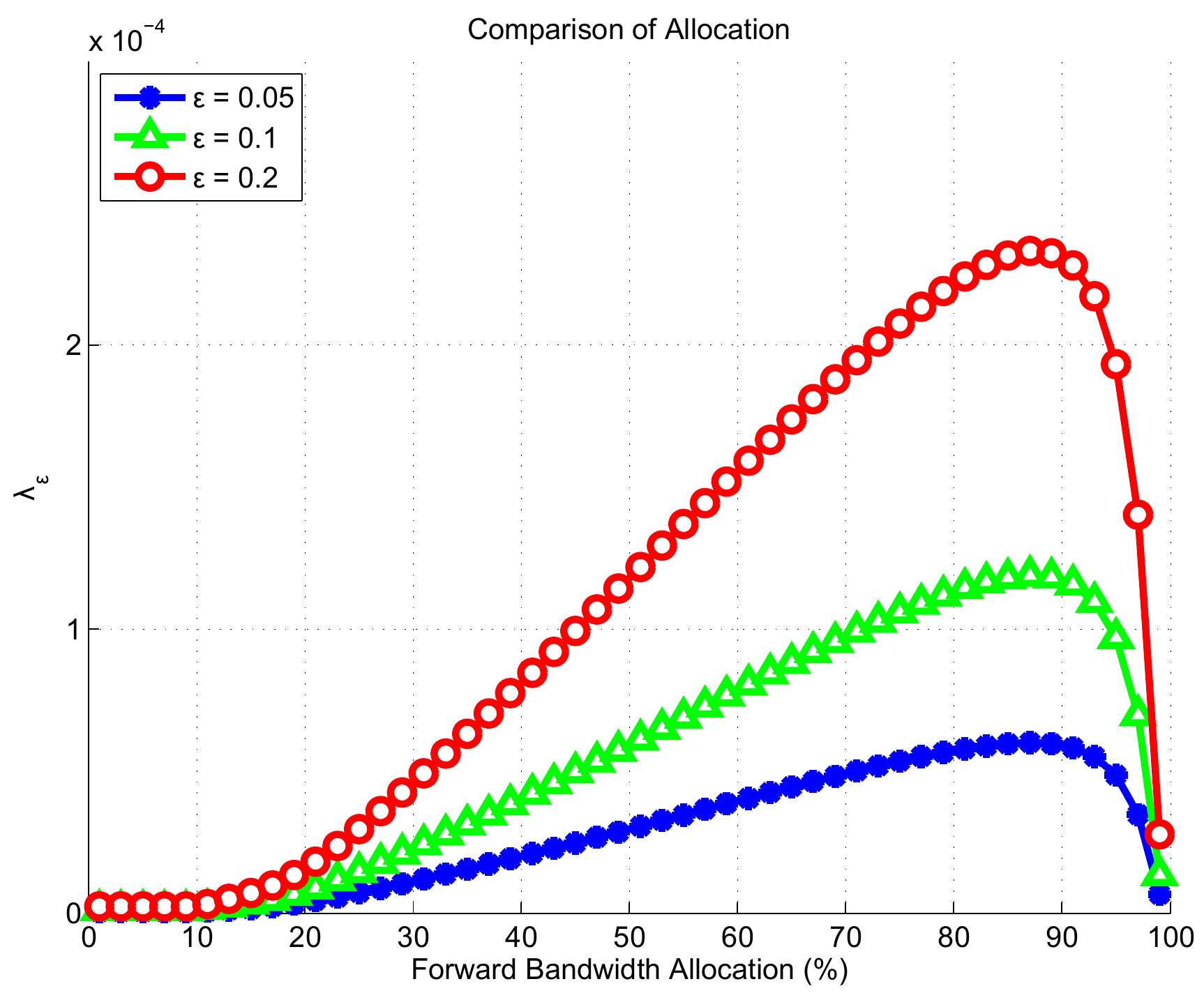}
      	\caption{The transmission capacity of a two-way network can be improved by allocating bandwidth in an optimal way.}
 \label{fig:opt_allocation}
\end{figure}
      
Fig. \ref{fig:opt_allocation} shows the transmission capacity as a function of forward bandwidth allocation. As more bandwidth is added to the forward link, the required $\text{SINR}_\text{F}$ decreases to meet the rate requirement. Because the reverse link rate requirement is quite small, the increase in $\text{SINR}_\text{R}$ does not change the SINR probability much (i.e. we are operating at very low $\text{SINR}_\text{R}$ which is where the SINR probability plateaus to 1). Fig. \ref{fig:opt_allocation} shows the naiveté of simply splitting the bandwidth in half. A nearly 2x improvement in transmission capacity is achieved by going from 50\% to the optimal allocation of 90\%. What is somewhat more surprising is that a 96\% split (i.e. splitting according to the rate requirement) results in nearly the same performance as a naive 50\% allocation.  Lastly, Fig. \ref{fig:opt_allocation} shows that this allocation is invariant to outage constraint. 

Fig \ref{fig:asym_ase} shows the performance gains in terms of area spectral efficiency that can be achieved by various bandwidth allocations. In all curves, the sum rate of the system is 208Mbps. As expected from Fig. \ref{fig:opt_allocation}, the area spectral efficiency is the worst in the naive 50/50 bandwidth allocation. The rate based (96\%/4\%) allocation performs better, but additional gains can be made by further optimizing the allocation. With the optimal allocation, the two-way system can achieve 75\% the area spectral efficiency of the one-way system. Because the one-way and two-way area spectral efficiency is linear in $\lambda_\epsilon$ and $\lambda_\epsilon^{\text{tw}}$, respectively, we can see the effect two-way communication has on the transmission capacity. If the users split the resources equally, considering the two-way constraint reduces the density by nearly a factor of 3. If the resources are split optimally, the network can support 2$\times$ the number of users from the equal split. This density is roughly 75\% of one-way density.

\begin{figure}
      \centering
      \includegraphics[width=0.65\columnwidth]{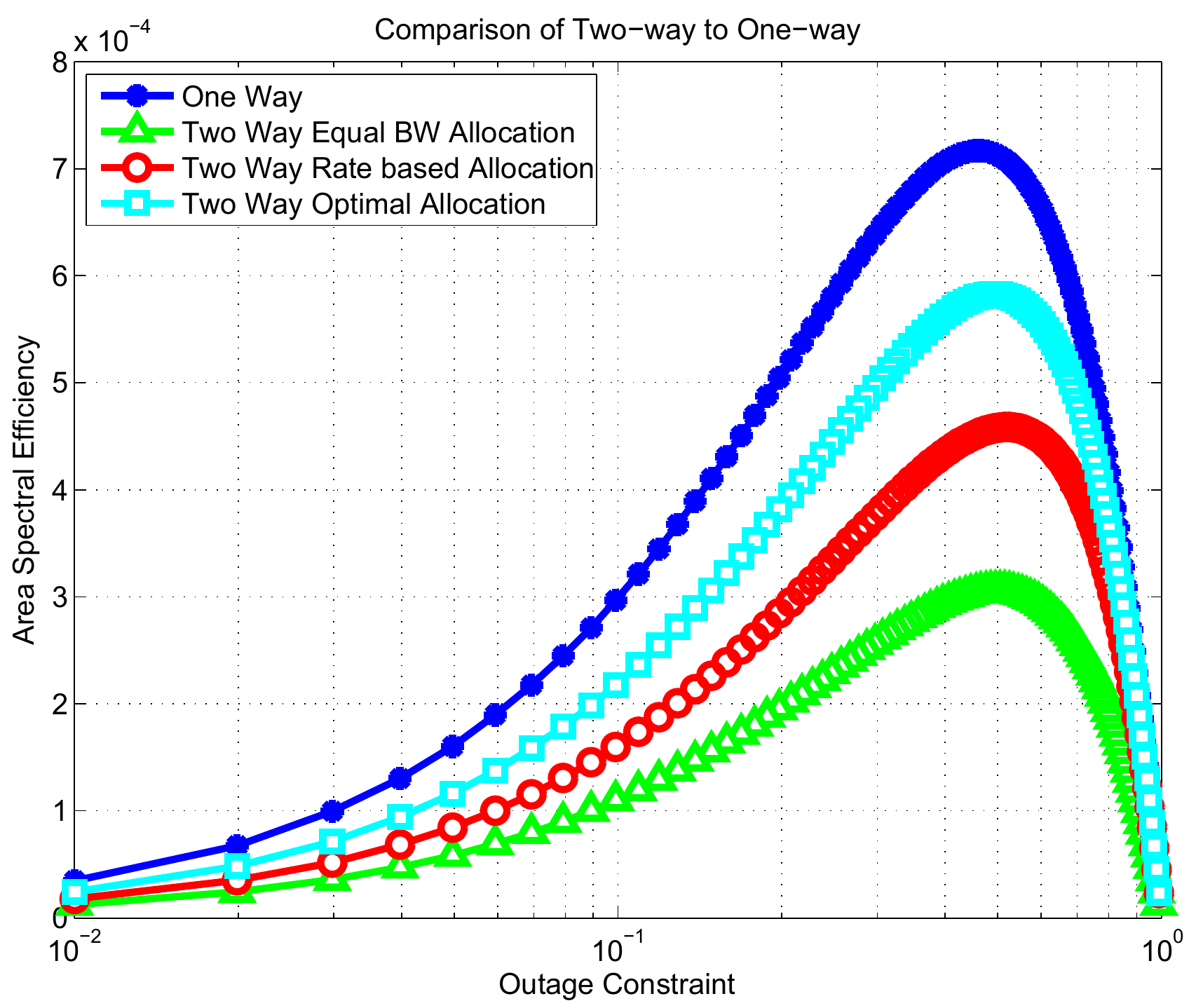}
      	\caption{Significant ASE gains can be achieved by intelligently allocating bandwidth.}
 \label{fig:asym_ase}
\end{figure}

\section{Conclusions}

%We presented an analysis that characterized the performance of mmWave \emph{ad hoc} networks for both one-way and two-way communication. We analyzed the coverage probability, showing that mmWave \emph{ad hoc} networks have quite favorable SINR distributions even in fairly dense networks. We proposed several reasonable simplifications in order to derive outage-optimal network densities, area spectral efficiency, and the rate coverage of mmWave \emph{ad hoc} networks. As the next-generation of wireless networks are developed, there are many candidate technologies. We believe that mmWave \emph{ad hoc} networks can help meet the anticipated 1000$\times$ increase in capacity that will be needed for D2D or emergency applications in the coming decade. 

We presented an analysis that characterized the performance of mmWave \emph{ad hoc} networks for both one-way and two-way communication. We showed that mmWave networks can improve on the performance and efficiency of UHF networks when considering both LOS and NLOS communication. Massive improvements in transmission capacity and area spectral efficiency (e.g. 10-100$\times$) are possible when only communicating over LOS links which motivates LOS aware protocols. Further, we showed the NLOS interference is negligible and LOS interference can still be the limiting factor for a mmWave \emph{ad hoc} network. This also motivates the need for LOS interference mitigation strategies. Lastly, by, understanding the requirements of the reverse link in the mmWave network for two way traffic, 75\% of the one-way capacity can be achieved which is twice as efficient as an equal allocation of resources.

%\section{Acknowledgments}
%The authors would like to acknowledge support from Army Research Labs under Grant No. W911NF-12-R-0011 and the National Science Foundation under Grant No. 1218338.

\appendix
\emph{Proof of Lemma 1}: From \cite{gamma} Theorem 1,
\begin{equation}
\left[1-e^{-\beta x^p}\right]^{1/p} < \frac{\int_{0}^{x}e^{-t^p}dt}{\Gamma(1+1/p)}
\end{equation}
with $\beta = \left[\Gamma(1+1/p)\right]^{-p}$ and $p \in (0,1)$. It is shown in \cite{gamma} that
\begin{equation}
\int_{0}^{x}e^{-t^p}dt = \frac{1}{p}\gamma \left(\frac{1}{p},x^p\right)
\end{equation}
where $\gamma(\cdot,\cdot)$ is the lower incomplete gamma function. A normalized gamma random, $y \sim \Gamma(k,\theta)$, variable is such that the shape, $k$, and scale, $\theta$, are inverses of each other so that $\bbE[y] = 1$ (i.e $\theta = 1/k$). If we let $k = 1/p$ and $x^p = kz$, we have
\begin{equation}
\begin{split}
\left[1-e^{-\beta x^p}\right]^{1/p} &<\frac{1}{p}\frac{\gamma \left(\frac{1}{p},x^p\right)}{\Gamma(1+1/p)}\\
\left[1-e^{-\beta kz}\right]^{k} &< \frac{k\gamma \left(k,kz\right)}{\Gamma(1+k)}\\
\left[1-e^{-az}\right]^{k}&< \frac{\gamma \left(k,kz\right)}{\Gamma(k)} \\
 &=\bbP[y < z]
\end{split}
\end{equation}
with $a = k\left[\Gamma(1+k)\right]^{-1/k} = k (k!)^{-1/k}$.
\bibliographystyle{IEEEtran}
\bibliography{IEEEabrv,Tianyang_ad_hoc,mac}

\end{document}